\newif\ifnotes
\notestrue 

\documentclass[conference]{IEEEtran}
\IEEEoverridecommandlockouts
\usepackage[usenames,dvipsnames,table]{xcolor}
\usepackage{multicol} 

\usepackage{wrapfig}
\usepackage{algpseudocode}

\usepackage[normalem]{ulem}
\usepackage{setspace}
\usepackage{siunitx} 
\usepackage{pgf}
\usepackage{tikz}
\usepackage{tikzscale}
\usetikzlibrary{arrows,automata, calc, decorations.pathreplacing, shapes,positioning}
\ifnotes
  \usepackage[nomargin,inline,draft]{fixme}
\else
  \usepackage[nomargin,inline,final]{fixme}
\fi

\usepackage[utf8]{inputenc}
\usepackage[T1]{fontenc}
\usepackage{microtype}
\usepackage{pstricks}
\usepackage{graphicx}
\usepackage{amsmath}
\usepackage{amssymb}
\usepackage{xspace}
\usepackage{listings}
\usepackage{multirow} 
\usepackage{amsthm}

\ifCLASSOPTIONcompsoc
    \usepackage[caption=false, font=normalsize, labelfont=sf, textfont=sf]{subfig}
\else
    \usepackage[caption=false, font=footnotesize]{subfig}
\fi

\newcommand{\OpenBox}[1][.4pt]{\leavevmode
  \hbox{
  \hfil\vrule width#1
  \vbox to.675em{\hrule width.6em height#1\vfil\hrule height#1}%
  \vrule width#1\hfil}}

\newtheorem{lemma}{Lemma}
\newtheorem{definition}{Definition}
\newtheorem{theorem}{Theorem}

\usepackage{mathtools}

\usepackage{float}

\usepackage{flushend}

\newcommand{\eqeq}{\hspace{1mm}\raisebox{0ex}{\relsize{0}{\textbf{==}}}~}

\usepackage[
  sorting=none,
  style=numeric-comp,
  backend=biber,
  isbn=false,
  url=true,
  doi=false,
  url=false,
  mincrossrefs=100,
  maxnames=12,
  giveninits=true
]{biblatex}

\AtEveryBibitem{\clearname{editor}}

\definecolor{vrpink}{RGB}{255,0,127}
\definecolor{vrblue}{RGB}{30,144,255}
\definecolor{vrolive}{RGB}{85,107,47}
\definecolor{vrroyalblue}{RGB}{65,105,225}
\definecolor{brgreen}{RGB}{100,200,70}
\definecolor{ivsalmon}{RGB}{255,160,122}
\definecolor{vrlpink}{RGB}{255,192,203}
\definecolor{mvcol}{RGB}{5,150,25}

\definecolor{jdgreen}{RGB}{4,99,7}
\definecolor{jdred}{RGB}{133,6,6}

\fxusetheme{color}
\fxuseenvlayout{color}

\usepackage{graphicx}
\graphicspath{{./figures/}}
\DeclareGraphicsExtensions{.pdf,.jpeg,.png,.jpg}

\usepackage[plain]{algorithm2e}
\usepackage{algorithmicx}
\usepackage{algpseudocode}
\usepackage{algcompatible}
\algblockdefx{Event}{EndEvent}[1]%
{\textbf{upon event }#1 \textbf{do }}%
{\textbf{}}

\algblockdefx{Function}{EndFunction}[1]%
{\textbf{Function }#1 }%
{\textbf{}}

\definecolor{mymaroon}{RGB}{128,0,0}
\newcommand{\comalgo}[1]{ \textit{\textcolor{mymaroon}{$\triangleleft$ #1}}}

\newcommand{\emptypath}{\emptyset_p}

\newcommand{\emptylist}{\emptyset_l}

\usepackage[
  colorlinks=true,
  linkcolor=purple,
  citecolor=purple,
  urlcolor=purple,
  pdfauthor={},
  pdftitle={},
  pdfsubject={},
  pdfkeywords={},
  bookmarks=false,
]{hyperref}

\hypersetup{colorlinks=true,urlcolor=purple,citecolor=purple,urlcolor=purple}

\usepackage{pifont}
\newcommand{\checkyes}{\textcolor{jdgreen}{\ding{51}}}
\newcommand{\checkno}{\textcolor{jdred}{\ding{55}}}

\newcommand{\nauth}{\mathit{nauth}}
\newcommand{\auth}{\mathit{auth}}

\makeatletter
\renewenvironment{proof}[1][\proofname]{\par
  \pushQED{\qed}%
  \normalfont
  \topsep0pt \partopsep0pt 
  \trivlist
  \item[\hskip\labelsep
        \itshape
    #1\@addpunct{.}]\ignorespaces
}{%
  \endtrivlist\@endpefalse
}
\makeatother



\usepackage{mathtools}
\newtagform{blue}{\color{BlueViolet}(}{)}

\newcommand{\du}{\texttt{DolevU}\xspace} 
\newcommand{\dut}{\texttt{DolevU-T}\xspace} 
\newcommand{\dutverif}{\texttt{DolevU-T-Verif}\xspace} 
\newcommand{\sft}{\texttt{SigFlood-T}\xspace} 
\newcommand{\sftverif}{\texttt{SigFlood-T-Verif}\xspace}
\newcommand{\dualrc}{\texttt{DualRC}\xspace} 
\newcommand{\dualrcverif}{\texttt{DualRC-Verif}\xspace}

\DeclareUnicodeCharacter{0301}{\'{e}}

\usepackage[export]{adjustbox}

\newtheorem{mydef}{Definition}

\begin{document}

\author{
    \IEEEauthorblockN{Rowdy Chotkan\IEEEauthorrefmark{1}, Bart Cox\IEEEauthorrefmark{1}, Vincent Rahli\IEEEauthorrefmark{2}, J\'er\'emie Decouchant\IEEEauthorrefmark{1}}
    \IEEEauthorblockA{\IEEEauthorrefmark{1}\textit{Delft University of Technology}, \IEEEauthorrefmark{2}\textit{University of Birmingham} \\
    \{R.M.Chotkan-1, B.Cox, J.Decouchant\}@tudelft.nl, vincent.rahli@gmail.com
    }
}



\title{Reliable Communication in Hybrid Authentication and Trust Models}

\maketitle
\IEEEpeerreviewmaketitle

\thispagestyle{plain}
\pagestyle{plain}

\begin{abstract}
Reliable communication is a fundamental distributed communication abstraction that allows any two nodes of a network to communicate with each other. It is necessary for more powerful communication primitives, such as broadcast and consensus. Using different authentication models, two classical protocols implement reliable communication in unknown and sufficiently connected networks. In the first one, network links are authenticated, and processes rely on dissemination paths to authenticate messages. In the second one, processes generate digital signatures that are flooded in the network. This work considers the hybrid system model that combines authenticated links and authenticated processes. We additionally aim to leverage the possible presence of trusted nodes and trusted components in networks, which have been assumed in the scientific literature and in practice. We first extend the two classical reliable communication protocols to leverage trusted nodes. We then propose \dualrc, a novel algorithm that enables reliable communication in the hybrid authentication model by manipulating both dissemination paths and digital signatures, and leverages the possible presence of trusted nodes (e.g., network gateways) and trusted components (e.g., Intel SGX enclaves). We provide correctness verification algorithms to assess whether our algorithms implement reliable communication for all nodes on a given network.
\end{abstract}

\section{Introduction}

Distributed systems involve autonomous computing processes (also called nodes) that communicate using a network to perform a cooperative task. Formalization efforts have led to the definition of distributed computing abstractions that precisely state a set of properties that algorithms must provide~\cite{Cachin+Guerraoui+Rodrigues:2011}. In particular, fault-tolerant distributed computing abstractions tolerate a limited number of faulty nodes. The most general fault model is the Byzantine model, which allows a process to deviate from a specified protocol in unrestricted ways and makes them apt to deal with the real world. For example, protocols that implement the Byzantine consensus abstraction are involved in Blockchain systems~\cite{ decouchant2022damysus,Yin+Malkhi+Reiter+Gueta+Abraham:podc:2019}. 

Distributed computing abstractions can be built on top of each other, which simplifies implementations and allows for simpler proofs of correctness.  
Arguably, the simplest distributed computing abstraction is reliable communication~\cite{dolev1981unanimity, koo2004broadcast, pelc2005broadcasting} because it allows two processes to authenticate each other's messages. In an unknown and incomplete network, reliable communication relies on broadcast and requires that: (i) when a correct process broadcasts a message, then the message is authenticated (or delivered) by all correct nodes; and (ii) when a message originating from a correct process is delivered, it was indeed sent by this correct process. Reliable communication is tightly related to reliable broadcast~\cite{Bracha:iandc:1987} and provides the same properties when the original sender is correct, which leads to reliable communication being sometimes called reliable broadcast with an honest dealer.

State-of-the-art reliable communication protocols have either considered that all nodes are authenticated (i.e., can generate and verify signatures) or, instead, that all nodes rely on authenticated links~\cite{dolev1981unanimity}. However, because of the recent development of complex systems that involve a large variety of nodes, some using digital signatures and some using point-to-point authenticated communications, there is a need to revisit the reliable communication abstraction and assume hybrid models. 

Several academic works have considered hybrid process models and additional trust assumptions, i.e., trusted nodes~\cite{tseng2019reliable, tseng2020reliable} or trusted components~\cite{correia2002efficient,decouchant2022damysus, Kapitza+al:eurosys:cheapbft:2012,Veronese+al:tc:MinBFT:2013}, mostly for higher-level distributed computing abstractions such as consensus. The impact of trusted nodes and trusted components on lower-level distributed computing abstractions, such as reliable communication, has not been evaluated. 
In this work, we revisit state-of-the-art reliable communication protocols and extend them to tolerate the simultaneous presence of authenticated and non-authenticated nodes and leverage the possible presence of trusted nodes or of nodes equipped with a trusted component. While our algorithms cover a large range of possible networks that will require extensive theoretical analysis, we mostly aim to enforce reliable communication in regular hierarchical network structures. 
Tab.~\ref{tab:protocols} lists the state-of-the-art RC protocols that we consider in this paper, along with the three novel protocols that we introduce, and compares their system models and correctness conditions.

\renewcommand{\arraystretch}{1.1}
\begin{table*}[t]
\footnotesize
\centering
\resizebox{1.6\columnwidth}{!}{%
\begin{tabular}{|l|c|c|c|c|m{5cm}|}
\hline
\multicolumn{1}{|l|}{\textbf{Protocol}} & \textbf{\begin{tabular}[c]{@{}c@{}}Auth.\\links\end{tabular}} & \textbf{\begin{tabular}[c]{@{}c@{}}Auth.\\nodes\end{tabular}} &\textbf{\begin{tabular}[c]{@{}c@{}}Trust.\\nodes\end{tabular}}  & \textbf{\begin{tabular}[c]{@{}c@{}}Trust.\\cmpts\end{tabular}} & \textbf{\begin{tabular}[c]{@{}c@{}}Correctness condition\end{tabular}} \\
\hline
DolevU~\cite{dolev1981unanimity} & \checkyes & \checkno & \checkno & N/A & $2f{+}1$ vertex-connectivity  \\
\hline 
SigFlood~\cite{eugster2004epidemic} & \checkno & \checkyes & \checkno & \checkno & $f{+}1$ vertex-connectivity \\
\hline
\hline
DolevU-T (\S\ref{sec:dut}, Alg.~\ref{alg:dolev_u_t}) & \checkyes & \checkno & \checkyes & N/A & \dutverif{} (\S\ref{sec:dolev_u_t_verif_2} and \S\ref{sec:dolev_u_t_verif_1}, Alg.~\ref{alg:dolev_u_t_verif_2} or Alg.~\ref{alg:dolev_u_t_verif_1}) \\
\hline 
SigFlood-T (\S\ref{sec:sigflood_t}, Alg.~\ref{alg:floodingtrusted}) & \checkno & \checkyes & \checkyes & \checkyes & \sftverif{} (\S\ref{sec:sigflood_t_verif}, Alg.~\ref{alg:dolev_u_t_verif_2} or Alg.~\ref{alg:dolev_u_t_verif_1} with $2f{+}1$ replaced by $f{+}1$) \\
\hline 
DualRC (\S\ref{sec:dualrc}, Alg.~\ref{alg:dualrc}) & \checkyes & \checkyes & \checkyes & \checkyes & \dualrcverif{} (\S\ref{sec:dualrc})  \\
\hline
\end{tabular}
}
\vspace{2pt}
\caption{Reliable communication protocols, their system models, and correctness conditions. The three detached bottom rows are novel protocols discussed in this paper.}
\label{tab:protocols}
\end{table*}

\renewcommand{\arraystretch}{1.0}

\smallskip
\noindent In summary, this work makes the following contributions: \\
$\bullet$ We extend the state-of-the-art reliable communication protocols to leverage trusted nodes (but \emph{not} trusted components) in unknown networks under the global fault model. First, we modify Dolev's~\cite{dolev1981unanimity} reliable communication protocol, which disseminates messages across authenticated links along with their propagation paths, facilitating message authentication (\S\ref{sec:dut}). Second, we modify the flooding-based reliable communication protocol~\cite{eugster2004epidemic}---referred to as \texttt{SigFlood}. This protocol, which relies on authenticated nodes, ensures that the dissemination of a message is accompanied by the digital signature of the original sender (\S\ref{sec:sigflood_t}).  \\ 
$\bullet$ Building on these insights, we then present \dualrc: the first reliable communication protocol that supports the presence of both authenticated and non-authenticated processes, and leverages trusted nodes and trusted components (\S\ref{sec:dualrc}). This protocol disseminates messages that contain dissemination paths, signed payloads and signed dissemination paths. \dualrc{} aims at supporting the largest possible range of network topologies where all authenticated and non-authenticated nodes can communicate reliably. Interestingly, \dualrc leverages the presence of trusted components to translate signed dissemination paths into trusted signatures that vouch for the authenticity of a message.  Fig.~\ref{fig:motivation_example} shows two networks to illustrate that \dualrc implements reliable communication in more generic networks than previous algorithms. \\ 
$\bullet$ For each of our protocols, we establish the necessary and sufficient conditions that the network topology must meet for all three  protocols to enforce all properties of the reliable communication distributed abstraction. Additionally, we present algorithms that allow for the verification of these conditions and evaluate their complexity. 


\begin{figure}[t]
\centering
     \subfloat[] {
         \centering
         \includegraphics[width=.45\columnwidth]{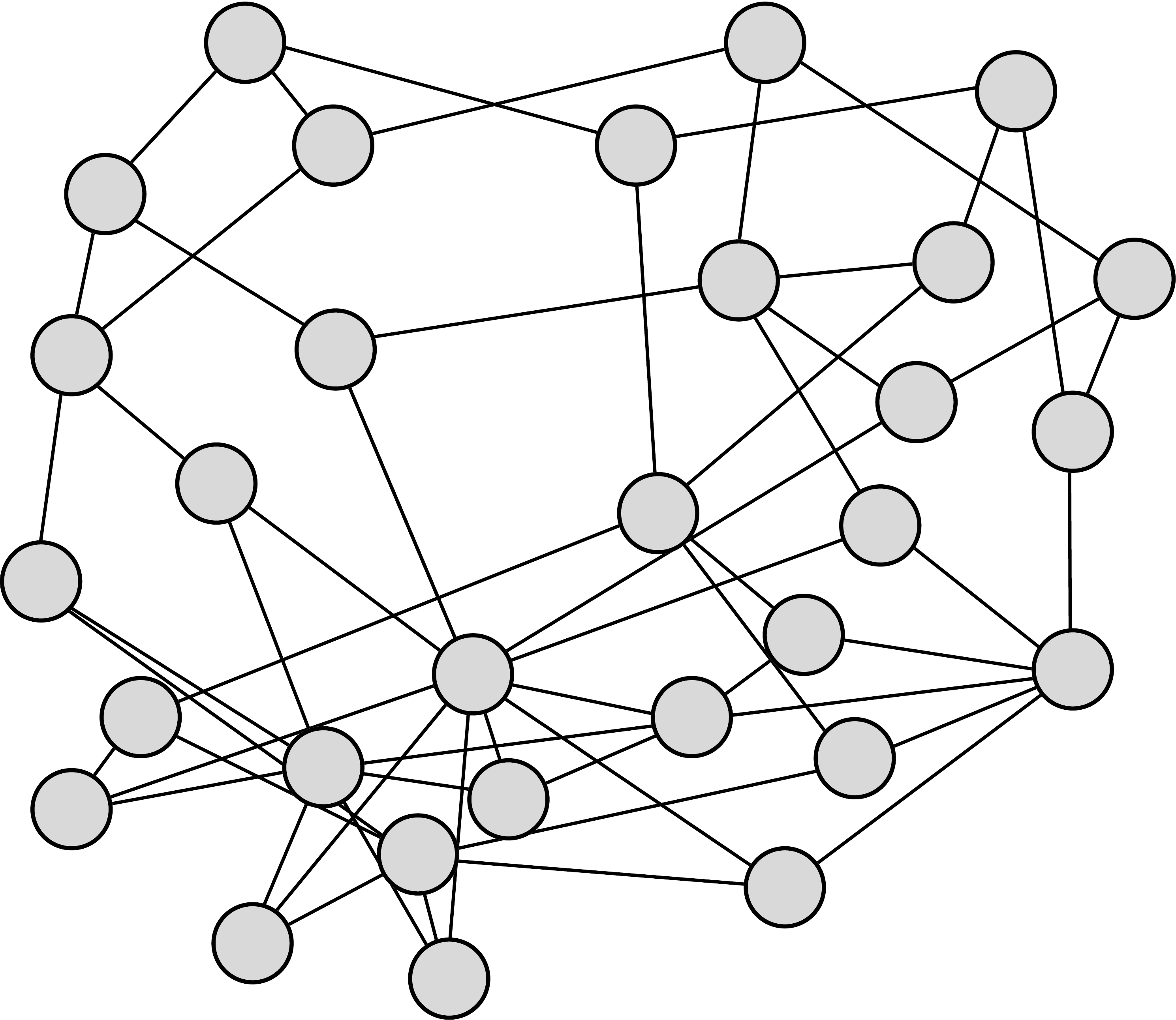}
     }
     \subfloat[] {
         \centering
         \includegraphics[width=.45\columnwidth]{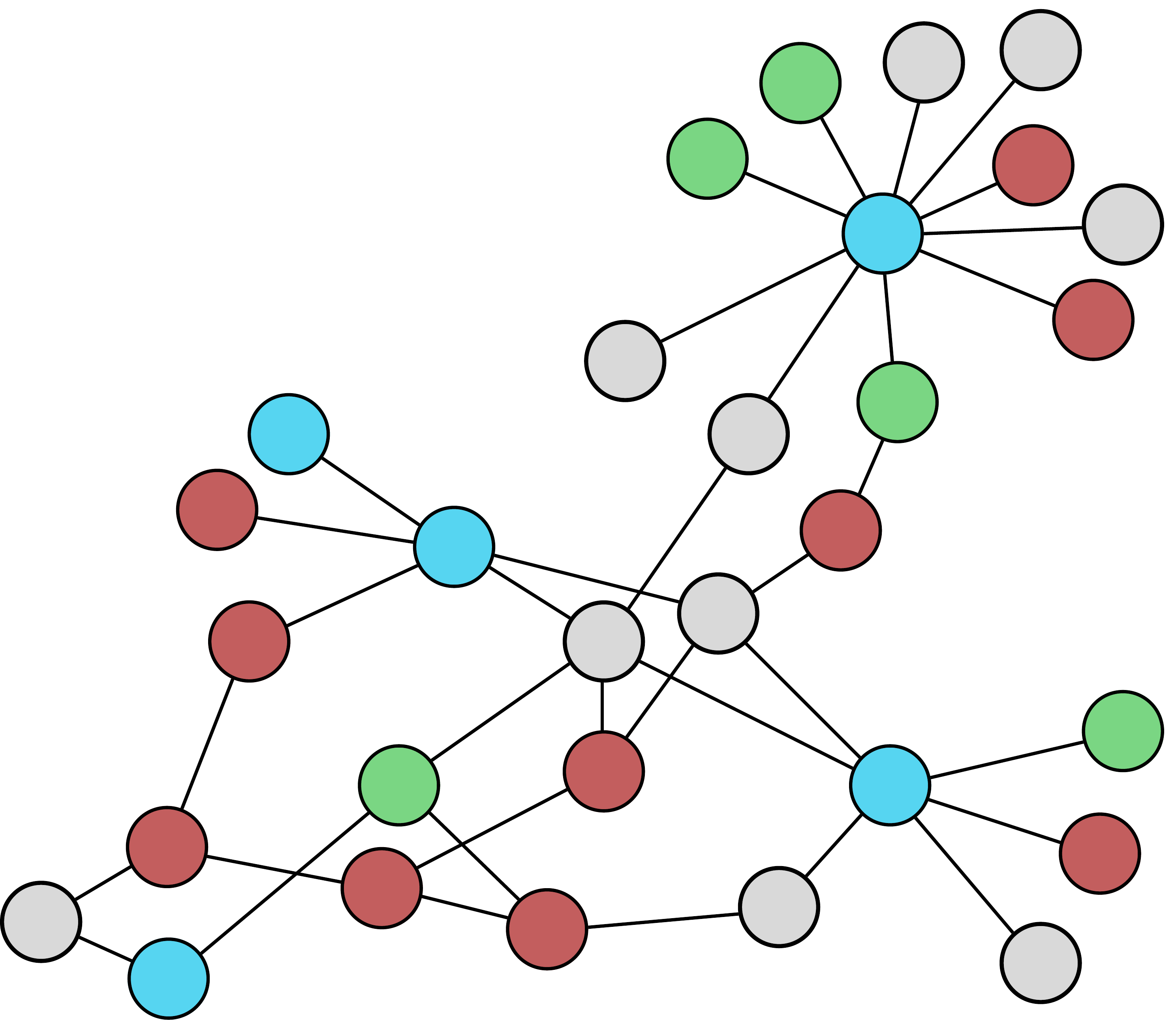}
     }
    \caption{The network shown in (a) contains 30 authenticated nodes with sufficient connectivity (3) for both SigFlood and \dualrc{} to be correct when $f=1$. The network illustrated in (b) contains 10 authenticated nodes (in red), 10 non-authenticated nodes (in grey), 5 authenticated trusted nodes (in green), and 5 non-authenticated trusted nodes (in blue). \dualrc{} is the only protocol that guarantees reliable communication in network (b).}
    \label{fig:motivation_example}

\end{figure}

\section{Models and Problem Statement}
\label{sec:sysmodel}
We consider a static set $V {=} \{p_1, p_2, \dots, p_N\}$ of $N$ processes that are identified by an integer in $[1,N]$. Processes are assumed to know the value of $N$.

\textbf{Network.} Processes are interconnected by a communication network, depicted as an undirected graph $G {=} (V,E)$ where each node $p_i {\in} V$ corresponds to a process, and each edge $e_{ij} {\in} E$ denotes a unicast communication channel between two processes. In this paper, we use `process' and `node' interchangeably, i.e., we assimilate a process to the network vertex that hosts it. We assume that the network topology is not known to nodes and remains static. However, processes know their direct neighbors and their identity. Two nodes that are not directly connected by an edge $e_{ij} {\in} E$ have to rely on other, potentially Byzantine, nodes to relay their messages.
Communication channels are authenticated, ensuring that only the two nodes at their extremities can use them to transmit messages. The channels are also reliable, meaning that messages are not altered or lost during transmission. Finally, the channels are asynchronous, with unbounded transmission delays.
Tab.~\ref{tab:notations} summarizes the graph-related notations we use in this paper.

\begin{table*}
\centering
\footnotesize
\begin{tabular}{|r|p{.8\textwidth}|}
\hline
\multicolumn{1}{|r|}{\textbf{Notation}} & \multicolumn{1}{|l|}{\textbf{Definition}} \\
\hline
$G$ & An undirected graph. \\
$\Gamma_G(u)$ & The neighbors of node $u$ in graph $G$. \\
$\kappa_G(u,v)$ & Vertex-connectivity between two nodes $u$ and $v$ in graph $G$. \\
$f(G)$ & The graph whose vertices are $G$'s untrusted nodes and whose edges connect the untrusted nodes of $G$ that are neighbors in $G$ or connected by a path of trusted nodes in $G$ (cf. Def.~\ref{def:fg}). \\
$g(G,f(G),u)$ & The graph $f(G)$ to which node $u \in G$ has been added along with its edges (direct connections with nodes in $f(G)$, or connection to them through a path of trusted nodes in $G$) (cf. Def.~\ref{def:gfg}). \\ 
\hline
$\emptypath$ & An empty path. \\
$\emptylist$ & An empty list. \\
$\sigma_i(m,p_s)$ & Signature of process $p_i$ on payload data. $m$ broadcast by process $p_s$. \\
$\sigma_i(pa, m, p_s)$ & Signature of process $p_i$ on a dissemination path $pa$, and  payload data $m$ broadcast by process $p_s$. \\
$\sigma_{{TC}_i}(m,p_s)$ & Signature of process $p_i$'s trusted component on payload data $m$ broadcast by process $p_s$. \\
$\texttt{isTrusted}(p_i)$ & Utility function that indicates whether process $p_i$ is trusted. \\
$\texttt{isTC}(\sigma_i(m, p_s))$ & Utility function that indicates whether a signature $\sigma_i(m, p_s)$ was generated by process $p_i$'s trusted component on payload data $m$ broadcast by $p_s$. \\
\hline
\end{tabular}
\vspace{2pt}
\caption{Notations and definitions.}
\label{tab:notations}
\end{table*}

\textbf{Channels interface.} We assume that
processes have access to a virtual interface $\mathit{al}$ that
abstracts their communication channels\footnote{$\mathit{al}$ is an
acronym for \emph{authenticated channel}.}. The virtual interface
$\mathit{al}$ exposes two functions, $\texttt{Send}$ and
$\texttt{Receive}$. A node $p_i$ can send a message $m$ to a
neighbor~$p_j$ by calling function $\langle\mathit{al},\texttt{Send} {\mid} {p_j,m}\rangle$. Similarly, an incoming message $m$ sent by a neighbor $p_j$ triggers a function $\langle\mathit{al},\texttt{Receive}{\mid}{p_j,m}\rangle$ that node $p_i$ executes. 

\textbf{Objective.} We aim at implementing the reliable communication
(RC) abstraction. RC allows any two correct processes to exchange
messages and authenticate them. The interface of an RC protocol contains two functions, $\langle\texttt{RC-broadcast} {\mid} {m}\rangle$ and
$\langle\texttt{RC-deliver} {\mid} {m}\rangle$, that respectively broadcast a message $m$ and transfer a message $m$ to a higher-level application. In an unknown network, RC is implemented via broadcasting to ensure that a message reaches all nodes, including its intended recipient who is the only node to deliver it. We call \emph{broadcaster} the initial sender of a message. Def.~\ref{def:rc} defines the RC distributed computing abstraction formally.

\begin{mydef} \label{def:rc}
A reliable communication protocol ensures that the following three properties hold:\\
\textbf{RC-No duplication:} A correct process RC-delivers a message at most once. \\
\textbf{RC-No creation:} If a correct process RC-delivers a message, then it was RC-broadcast. \\
\textbf{RC-Validity:} If a correct process RC-broadcasts a message, then all correct processes eventually RC-deliver it. 
\end{mydef}

\textbf{Node categories.} Fig.~\ref{fig:venn_diagram} illustrates the node categories we consider. We divide the set $V$ of the nodes into four subsets ($V_{auth}$, $V_{n.auth}$, $V_{t}$, and $V_{tc}$) whose composition is known to all nodes and that are defined as follows. 

The set $V$ of all nodes is divided into two complementary sets: $V_{auth}$ and $V_{n.auth}$. The former, $V_{auth}$, contains authenticated nodes capable of manipulating digital signatures, while the latter, $V_{n.auth}$, includes non-authenticated nodes that cannot use digital signatures. Non-authenticated nodes, therefore, rely entirely on authenticated links to authenticate the messages they receive and allow others to authenticate their messages.

We define $V_t {\subseteq} V$ as the set of trusted nodes, which are assumed to be correct. Trusted nodes always strictly follow the protocol specification they are given~\cite{abbas2017improving,bhawalkar2015preventing,navon2020mixed,tseng2019reliable}. They are, however, not necessarily authenticated. 

Finally, we define $V_{tc} {\subseteq} V$ as the set of nodes that are equipped with a trusted component (TC). Note that $V_t {\cap} V_{tc}$ is not necessarily empty, which means that a trusted node might also be equipped with a trusted component. TCs are able to store some data and execute code within a secluded area that their host operating system cannot access. TCs, however, rely on their host to interact with the network and to be provided with CPU cycles. TC-augmented nodes are strictly considered to be authenticated, i.e., $V_{tc} {\subset} V_{\auth}$. This is in line with modern TCs, which are usually authenticated using a remote attestation procedure that employs cryptographic
primitives~\cite{menetrey2022attestation}. 

\textbf{Fault model.}
We assume that up to $f$ out of the $N$ nodes are Byzantine. Faulty nodes
exclusively belong to $V {\setminus} V_{t}$. A faulty node might
deviate from a protocol specification in unrestricted ways but cannot break cryptographic primitives. For example, a faulty node can drop or modify messages it is expected to send. When equipped with a TC, a faulty node might not interact with it as intended.

\begin{figure}[tbp]
    \centering
    \includegraphics[width=0.7\columnwidth]{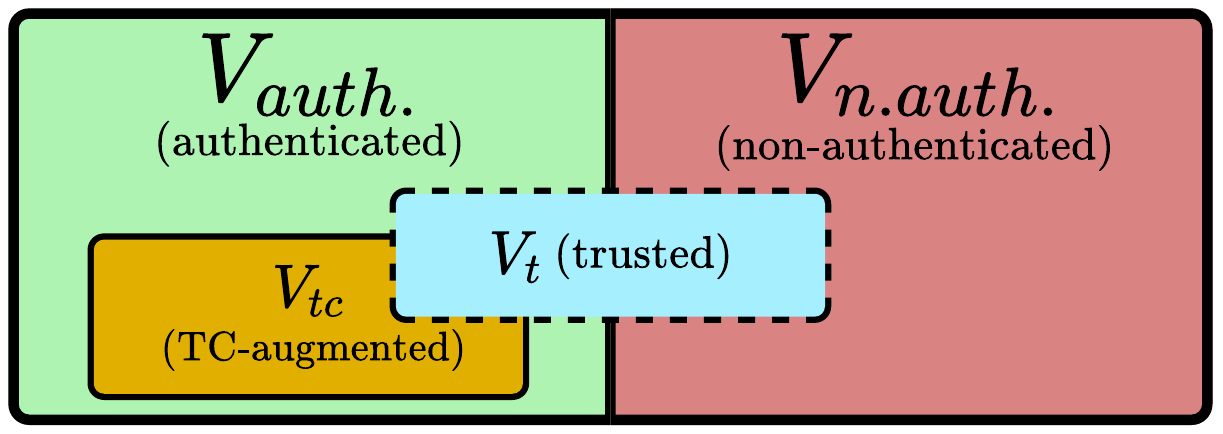}
    \caption{Taxonomy of processes.}
    \label{fig:venn_diagram}
\end{figure}

\textbf{Trusted component.} In our final algorithm, \dualrc{}, a trusted component can be used to do two things. First, it can verify a broadcaster's signature on a payload and generate its own signature to attest that it authenticated the payload. Second, it can receive $f{+}1$ signed dissemination paths that relate to a given payload as input, verify that the paths are vertex-disjoint, and then generate a signature to attest that the payload has been authenticated based on signed paths. We do not formally define our trusted component for space reasons but it could be implemented with existing trusted execution environments (e.g., with Intel SGX).   

\section{\dut: Reliable Communication in the Authenticated Link Model with Trusted Nodes}
\label{sec:dut}

We first consider the case where nodes only rely on authenticated links and do not use digital signatures, i.e., $V_{\auth} {=} \emptyset$ and $V_{n.auth} {=} V$. In this model, $V_{tc} {=} \emptyset$ because trusted components typically require digital signatures, but $V_t \subseteq V_{n.auth}$ is not necessarily empty (i.e., some nodes might be trusted).  

We describe our first protocol, \dut{}, which extends \du, the state-of-the-art reliable communication protocol with authenticated links, to leverage trusted nodes.
Interestingly, \du requires the network to be $(2f{+}1)$-connected, while \dut{} functions also correctly on more generic networks in the presence of trusted nodes, as App.~\ref{appx:example_dut} illustrates with an example. We, therefore, characterize the networks on which \dut{} is correct and provide two algorithms that take a network topology as input and indicate whether \dut{} allows any pair of nodes to receive and authenticate each other's messages (i.e., whether it enforces RC). 

\subsection{Description of \dut{}}

\RestyleAlgo{boxruled}
\LinesNumbered
\begin{algorithm}
\begin{algorithmic}[1]
\footnotesize 	
\State \textbf{Parameters:}
\State \hspace*{1em} $f$ : max. number of Byzantine processes in the system.
\State \textbf{Uses:} 
\State \hspace*{1em} $\bullet$ Auth. async. perfect point-to-point links, instance \textit{al}. 
\State \hspace*{1em} $\bullet$ $\Gamma_G(p_j)$ returns the list of $p_j$'s neighbors. 
\item[]

\State \underline{\textbf{upon event} $\langle \dut{}, \texttt{Init} \rangle$ \textbf{do}}
    \State \hspace*{1em} $delivered = \textbf{False}$ \comalgo{to deliver at most once}
    \State \hspace*{1em} $paths = \emptyset$ \comalgo{received dissemination paths}
\item[]

\State \underline{\textbf{upon event} $\langle \dut{}, \texttt{Broadcast}\ |\ m \rangle$ \textbf{do}}
    \State \hspace*{1em} \textbf{forall} $p_j \in \Gamma_{G}(p_i)$ \textbf{do} 
        \State \hspace*{2em} $\langle al, \texttt{Send}\ |\ p_j, [m, \emptypath] \rangle$ \label{alg:dolev-init-msg}
    \State \hspace*{1em} $\langle \dut{}, \texttt{Deliver}\ |\ m \rangle$
    \State \hspace*{1em} $delivered = \textbf{True}$
\item[]

\State \underline{\textbf{upon event} $\langle al, \texttt{Receive}\ |\ p_j, [m, path] \rangle$ \textbf{do}} \label{deliverMsg1}
    \State \hspace*{1em}  $fwd\_path = path + [p_j]$ \label{alg:dolev-augment-path}
    \State \hspace*{1em}  \textcolor{blue}{$path = \texttt{remove\_all\_trusted}(fwd\_path)$} \label{line:simplify}
    \State \hspace*{1em}  $paths.\texttt{insert}(path \setminus \{p_j\})$ \label{alg:dolev-insert-path}
    \State \hspace*{1em} \textbf{forall} $p_k \in \Gamma_G(p_i) \setminus path$ \textbf{do} \label{alg:dolev-forward}
        \State \hspace*{2em} $\langle al, \texttt{Send}\ |\ p_k, [m, fwd\_path] \rangle$ \label{line:forward}
\item[]

\State \underline{\textbf{upon event} ($paths$ contains $[\ ]$ (empty path) } \underline{\textbf{or} $f{+}1$ vertex-disjoint paths)} \underline{\textbf{and} $(\textbf{not}\ delivered)$ \textbf{do}} \label{alg:dolev-disj-paths}
    \State \hspace*{1em} $\langle \dut{}, \texttt{Deliver}\ |\ m \rangle$ 
    \State \hspace*{1em} $delivered = \textbf{True}$
\end{algorithmic}
\vspace{2mm}
\caption{\dut: Reliable communication in $(2f{+}1)$-connected networks at process $p_i$ with trusted processes in the authenticated link model.} 
\label{alg:dolev_u_t}
\end{algorithm}

Alg.~\ref{alg:dolev_u_t} presents \dut{}'s pseudocode.
Messages contain two fields: the payload data $m$ that is being broadcast and the dissemination path $path$.
To broadcast, a node sends a given payload with an empty path to its
neighbors (Alg.~\ref{alg:dolev_u_t}, l.~\ref{alg:dolev-init-msg}), and
delivers it. Upon receiving a message, a process $p_i$ augments the received path variable with the sender's ID
(Alg.~\ref{alg:dolev_u_t}, l.~\ref{alg:dolev-augment-path}). It then forwards the message to its neighbors that are not included in the received path
(Alg.~\ref{alg:dolev_u_t}, l.~\ref{line:forward}). A process stores the
paths it receives (Alg.~\ref{alg:dolev_u_t},
l.~\ref{alg:dolev-insert-path}) and delivers a message if it was either received directly from the initial broadcaster or if it can identify $f{+}1$ vertex-disjoint
paths among those it received (Alg.~\ref{alg:dolev_u_t},
l.~\ref{alg:dolev-disj-paths}). Receiving a payload through $f{+}1$
vertex-disjoint paths authenticates it because there are at most $f$
faulty nodes in the network. This implies that at least one of the
$f{+}1$ disjoint paths through which the message traveled involved only
correct nodes. 

\dut{} removes trusted nodes
from a received path (Alg.~\ref{alg:dolev_u_t},
l.~\ref{line:simplify}), before storing it in a $paths$ set and
attempting to identify $f{+}1$ vertex-disjoint paths. One should note
that trusted nodes are not removed from the path that a node transmits
to its neighbors (Alg.~\ref{alg:dolev_u_t}, l.~\ref{line:forward}). Removing them from transmitted paths would increase the total number of messages sent: a node that receives a path might forward a message to all its neighbors not included in it (Alg.~\ref{alg:dolev_u_t}, l.~\ref{alg:dolev-forward}).

Removing trusted nodes from paths has interesting practical consequences. First, it reduces the size and the number of paths a node stores and manipulates.
This is beneficial given that identifying disjoint paths
has an exponential time complexity. Second, paths containing a shared
list of trusted nodes might become disjoint. Finally, it might also generate an empty path from a path that contains only trusted
nodes, which accelerates the delivery of a message through a protocol
optimization presented in the next section.

\subsection{Optimizations}
\label{sec:bonomi_modifications}

The worst-case complexity of the original \du protocol is high, both in terms of message and computational complexity. Bonomi et al. presented five modifications of \du, which we refer to as MD.1--5, that reduce the number of messages transmitted along with their size in practical executions~\cite{bonomi2019multi}:

\begin{description}
\item[MD.1] If a process $p$ receives a message directly from the source $s$, then $p$ directly delivers it.

\item[MD.2] If a process $p$ has delivered a message, then it discards all its related paths and relays the message with an empty path to all its neighbors.

\item[MD.3] A process $p$ relays a path related to a message only to the neighbors that have not delivered it (i.e., sent an empty path).

\item[MD.4] If a process $p$ receives a message with an empty path from a neighbor $q$, then $p$ stops relaying and analyzing any path for that message that contains $q$.

\item[MD.5] A process $p$ stops relaying further paths related to a message after it has delivered it and has forwarded it with an empty path. 
\end{description}

\dut{} can make use of
MD.1--5, but these modifications are left out of Alg.~\ref{alg:dolev_u_t} for
simplicity. They, however, appear in the pseudocode of \dualrc{}, our
final protocol. \dut{} also leverages additional optimizations that
respectively prevent duplicate transmissions of a given message on a link (MBD.1 in~\cite{bonomi2021practical}), and allow nodes to ignore messages that contain a \textit{superpath} of a previously received path ((MBD.10 in~\cite{bonomi2021practical}). These two optimizations are also not
shown in the pseudocode, as their implementation is straightforward.

\subsection{Proofs}

We start by proving that \dut enforces \emph{RC-no duplication} and \emph{RC-no creation}. 

\begin{lemma} \label{lemma:1}
\dut{} maintains \textbf{RC-no duplication} and \textbf{RC-no
  creation}.
\end{lemma}

\begin{proof}
\textit{RC-no duplication} is guaranteed by the use of the $delivered$ variable.
In \du, a node delivers a message when it receives it through $f{+}1$ vertex-disjoint paths, which guarantees that at least one of those paths only involves correct nodes. The presence of trusted nodes in a path does not affect whether all the nodes it contains are correct; therefore, trusted nodes can be removed from a path used to attempt to authenticate a message, which proves \textit{RC-no creation}. 
\hfill $\OpenBox[1pt]$
\end{proof}
\vspace{1mm}

Note that \dut{}, like \du{}, does not guarantee \textit{RC-Validity} on all graphs. While it is known that \du{} requires a $2f{+}1$-connected communication graph, \dut{}'s requirements are more complex. \S\ref{sec:dolev_u_t_verif_2} and~\ref{sec:dolev_u_t_verif_1} provide algorithms that determine whether \dut{} provides \textit{RC-Validity} on a given graph, and therefore implements reliable communication on it. 

We now describe two methods one can use on a given graph $G=(V,E)$ to validate whether \dut{} provides \textit{Validity}, and therefore implements reliable communication. 
Our correctness verification algorithms rely on two functions $f(\cdot)$ and $g(\cdot)$ that we first define intuitively. Function $f(G)$ takes a graph $G$ as input and returns a new graph containing only the untrusted nodes of $G$. Furthermore, an edge is added between any two untrusted nodes if they are directly connected as neighbors in $G$ or if they are connected by a path of trusted nodes in $G$. (Also see Def.~\ref{def:fg}). Function $g(G, G', u_t)$ inserts a trusted node $u_t \in V_t$ in a subgraph $G'$ of a graph $G$ and adds an edge in $G'$ between $u_t$ and a node of $G'$ if they are neighbors or if they are connected by a path of trusted nodes in $G$. 

\setlength\abovedisplayskip{0pt}

\begin{definition} \label{def:fg}
Let $f$ be the function that takes a graph $G=(V,E)$ as input and outputs
the graph $f(G){=}(V_{0}, E_{0})$ such that: \\
$\bullet$ $V_{0} = V_{n.auth} {\setminus} V_t$ \\
$\bullet$ 
$\begin{array}[t]{l}
   \hspace*{-0.05in}\forall (u,v) \in V_{0}^2,\ 
    \big( (u, v) \in E \big)  \lor \\
     \Big( \exists u_{t_0},\cdots,u_{t_n} \in V_t.
       (u, u_{t_0}) \in E \land (u_{t_n}, v) \in E \\
       \land \forall i < n. (u_{t_i}, u_{t_{i+1}}) \in E
        \Big{)} 
    \Leftrightarrow (u, v) \in E_{0} \\
\end{array}$

\end{definition}

\noindent In practice, the edges of a node $u \in V_0$ can be computed with a
breadth-first traversal in~$G$ starting from $u$ and following trusted
edges. However, a faster approach to determine~$E_0$ consists in
computing the connected components of the graph whose vertices are $G$'s vertices and whose edges are the edges from $G$ that connect two trusted nodes. In such a graph, two untrusted nodes
 belong in the same connected component iff. they are connected in $f(G)$.

The function $g(.)$ is defined as per Definition~\ref{def:gfg}.
\begin{definition}
Let $g$ be the function that takes as input a graph $G$, a subgraph
$G'{=}(V' {\subseteq} V, E' {\subseteq} V^2)$ of $G$, and a trusted
node $v_t {\in} V_t$, and outputs a graph $g(G, G', v_t) {=} (V_1,
E_1)$ such that:
\begin{itemize}
\item $V_1 = V' \cup \{v_t\}$
\item
  $\forall (u,v) \in E',\ (u, v) \in E_1$ \hfill (i.e., $E'\subseteq{E_1}$)
\item
$\begin{array}[t]{l}
  \hspace*{-0.05in}\forall u \in V', 
    \big( (u,v_t) \in E \big) \lor \\
    \Big( \exists u_{t_0},\cdots,u_{t_n} \in V_t,\ (u, u_{t_0}) \in E  \land (u_{t_n}, v_t) \in E \\
      \land \forall i < n, (u_{t_i}, u_{t_{i+1}}) \in E
       \Big{)}
    \Leftrightarrow (u,v_t) \in E_1
\end{array}$
\end{itemize}
\label{def:gfg}
\end{definition}

\subsubsection{Method 1: Max-Flow on Transformed Graph}
\label{sec:dolev_u_t_verif_2}

\begin{algorithm}[ht]
\caption{Max-flow-based verification of \dut's \emph{RC-Validity} on a graph (\dutverif{}, method 1).}
\label{alg:dolev_u_t_verif_2}

\begin{algorithmic}[1]
\footnotesize 	
\State \textbf{Inputs:}
\State \hspace*{1em} $G=(V,E)$: undirected network topology.
\State \hspace*{1em} $V_t \subset V$: set of trusted nodes in $G$.
\State \hspace*{1em} $f$: max. number of Byzantine nodes.
\State \textbf{Output:}
\State \hspace*{1em} A boolean that indicates whether \dut is correct on $G$. 
\item[]

    \State \comalgo{Split and add all nodes to new directed graph}
    \State $dG = \mathrm{DiGraph}(\emptyset)$ \comalgo{Init. empty unweighted directed graph} \label{l:dgBegin}
    \State \textbf{forall} $u \in V$ \textbf{do}
            \State \hspace*{.5em} $dG.\mathrm{add\_vertex}(2u)$ \comalgo{\emph{in} node}
            \State \hspace*{.5em} $dG.\mathrm{add\_vertex}(2u+1)$ \comalgo{\emph{out} node}
            \State \hspace*{.5em} \textbf{if} $u \in V_t$ \textbf{then}
                $dG.\mathrm{add\_weighted\_edge}((2u, 2u+1), 2f+1)$
            \State \hspace*{.5em} \textbf{else}
                $dG.\mathrm{add\_weighted\_edge}((2u, 2u+1), 1)$
            \label{l:dgEnd}
    \item[]

    \State \comalgo{Add directed weighted edges to directed graph}
    \State \textbf{forall} $(u,v) \in E$ \textbf{do} \label{l:dgEdgeBegin}
        \State \hspace*{.5em} \textbf{if} $u \in V_t$ \textbf{then}
            $dG.\mathrm{add\_weighted\_edge}(2u+1, 2v, 2f+1)$
        \State \hspace*{.5em} \textbf{else}
            $dG.\mathrm{add\_weighted\_edge}(2u+1, 2v, 1)$
        \State \hspace*{.5em} \textbf{if} $v \in V_t$ \textbf{then}
            $dG.\mathrm{add\_weighted\_edge}(2v+1, 2u, 2f+1)$
        \State \hspace*{.5em} \textbf{else}
            $dG.\mathrm{add\_weighted\_edge}(2v+1, 2u, 1)$ \label{l:dgEdgeEnd}
    \item[]

\State \comalgo{Find out if any nodes $u$ and $v$ can communicate reliably}
\State \textbf{forall} $u \in V$ \textbf{do}
    \State \hspace*{.5em} \textbf{forall} $v \in V$ s.t. $u < v$ \textbf{do} 

        \State \hspace*{1em} \comalgo{a) Are $u$ and $v$ connected by an undirected trusted path?} \label{l:startconnected}
        \State \hspace*{1em} \textbf{if} ($u \notin V_t \land v \notin V_t \land v \in \Gamma_{f(G)}(u)$)
        \State \hspace*{2.3em} \textbf{or} ($u \in V_t \land v \notin V_t \land v \in \Gamma_{g(G,f(G),u)}(u)$)
        \State \hspace*{2.3em} \textbf{or} ($u \notin V_t \land v \in V_t \land u \in \Gamma_{g(G,f(G),v)}(v)$)
        \State \hspace*{2.3em} \textbf{or} ($u \in V_t  \land v \in V_t \land u \in \Gamma_{g(G,g(G,f(G),u),v)}(v)$) \textbf{then}
            \State \hspace*{1.5em} \textbf{continue} \label{l:endconnected}

    \item[]
    \State \hspace*{1em} \comalgo{b) Otherwise, do they have enough vertex-disjoint paths?}
    
   \State \hspace*{1em} \textbf{if} \textbf{not} \texttt{isTrusted}(u) \textbf{then}
       \State \hspace*{1.5em} \textbf{forall} $w \in \Gamma_G(u)$ \textbf{do}
            \State \hspace*{2em}
            $dG.\mathrm{change\_edge\_weight}(2u+1, 2w, 2f+1)$ \label{line:modifyWeights}
    \item[]
    \State \hspace*{1em} \textbf{if} $\mathrm{max\_flow}(dG, 2u+1, 2v) < 2f{+}1$ \textbf{then}
        \State \hspace*{1.5em}
        \textbf{return} \textsc{False} \label{l:returnFalse}
    \item[]
    \State \hspace*{1em} \textbf{if} \textbf{not} \texttt{isTrusted}(u) \textbf{then}
       \State \hspace*{1.5em} \textbf{forall} $w \in \Gamma_G(u)$ \textbf{do}
            \State \hspace*{2em}
            $dG.\mathrm{change\_edge\_weight}(2u+1, 2w, 1)$ \label{line:resetWeights}

\item[]

\State \textbf{return} $\textsc{True}$
\end{algorithmic}
\end{algorithm}

We now describe our first method which verifies whether \dut{} ensures reliable communication between any two nodes on a given network topology. Alg.~\ref{alg:dolev_u_t_verif_2} describes the corresponding pseudocode. This algorithm leverages the fact that the maximum number of edge-disjoint paths between two nodes $u$ and $v$ is equal to the maximum flow between them. It then modifies a graph $G$ into a graph $dG$ where a sufficient maximum flow between two nodes establishes their ability to communicate reliably using \dut{} (i.e., its \textit{RC-Validity}).

As we are interested in the number of vertex-disjoint paths between nodes, we must first modify the input graph $G$. Indeed, max-flow algorithms require edges to have a capacity and may use an edge up to its capacity, but several flows might cross a vertex. We, therefore, transform the graph using the node-splitting technique~\cite{ford1962flows}, which leads to the creation of a directed graph $dG$. More precisely, each node $u \in G$ is first split into two nodes~$u_{\mathrm{in}}$ and~$u_{\mathrm{out}}$, which are connected by a directed edge $(u_{\mathrm{in}}, u_{\mathrm{out}})$ in $dG$. Then, an undirected edge $(u,v)$ of~$G$ leads to the creation of two directed edges $(u_{\mathrm{out}}, v_{\mathrm{in}})$ and $(v_{\mathrm{out}}, u_{\mathrm{in}})$ in $dG$. In addition, all directed edges of $dG$ are initially given a unitary capacity. We say that a node in $dG$ is trusted if it has been created by splitting a trusted node of $G$. The capacity of all edges in~$dG$ that originate from a trusted node is set to $2f{+}1$. 

\begin{theorem} \label{thm:4}
\dut{} ensures \textit{RC-Validity} in $G$ iff. any two nodes $u, v \in G$ are either neighbors in $G$ or if the flow from $u_{\mathrm{out}}$ to $v_{\mathrm{in}}$ is larger than $2f{+}1$ in $dG$.
\end{theorem}

\begin{proof}
Let us remark that the flow from $u_{\mathrm{out}}$ to $v_{\mathrm{in}}$ in $dG$ is the same as the flow from $v_{\mathrm{out}}$ to $u_{\mathrm{in}}$ by construction. We, therefore, focus on the case where node $u$ sends a message to node~$v$.

\emph{($\Rightarrow$ Part 1/2)} The first delivery condition is node $v$ receiving an empty path. This implies that $u$ and $v$ are neighbors in $G$ or connected through a non-empty path of trusted nodes in $G$. In the second case, the flow from $u_{\mathrm{out}}$ to $v_{\mathrm{in}}$ is at least $2f{+}1$ in $dG$ because each edge from a trusted node has capacity $2f{+}1$. 

\emph{($\Rightarrow$ Part 2/2)} Node $v$ can also deliver a message if it obtains $f{+}1$ vertex-disjoint paths after removing the trusted nodes from the paths it receives. In the presence of up to $f$ faulty nodes, there should be $2f{+}1$ paths between $u$ and $v$ in $G$ that become vertex-disjoint once their trusted nodes are removed. These $2f{+}1$ paths each have a capacity greater than or equal to $1$.  They might share edges that originate from a trusted node and have a capacity of $2f{+}1$, but not their other edges, which have a capacity of $1$. Nonetheless, the existence of these $2f{+}1$ paths implies that the flow from $u$ to $v$ is at least $2f{+}1$ in $dG$. 

\emph{($\Leftarrow$ Part 1/2)} If $u$ and $v$ are neighbors in $G$ then they can always communicate reliably.  

\emph{($\Leftarrow$ Part 2/2)} If there is a flow larger than $2f{+}1$ between $u_{\mathrm{out}}$ to $v_{\mathrm{in}}$ in $dG$ then $u$ and $v$ might first be connected by a path of trusted nodes in $G$. In that case, they can communicate reliably. Otherwise, there are $2f{+}1$ vertex-disjoint paths that may or may not contain trusted nodes connecting $u_{\mathrm{out}}$ to $v_{\mathrm{in}}$ in $dG$. These paths correspond to paths in $G$ since an $\mathrm{out}$ node can only be reached by its corresponding $\mathrm{in}$ node in $dG$. In the presence of $f$ faulty nodes, at least $f{+}1$ paths in $G$ will therefore be dissemination paths and allow $v$ to receive and authenticate a message sent by $u$.

\hfill $\OpenBox[1pt]$
\end{proof}
\vspace{1mm}

As shown by Thm.~\ref{thm:4}, Alg.~\ref{alg:dolev_u_t_verif_2} 
determines whether all pairs $(u,
v)$ of nodes in a graph $G$ can communicate reliably with \dut{}. The
algorithm first computes the directed graph $dG$ from the input graph
$G$. Each node~$u$ is split into nodes $u_{\mathrm{in}} = 2u$ and
$u_{\mathrm{out}} = 2u+1$, which are connected with edge
$(u_{\mathrm{in}}, u_{\mathrm{out}})$ of capacity $2f{+}1$ if $u$ is
trusted or $1$ otherwise (Alg.~\ref{alg:dolev_u_t_verif_2},
ll.~\ref{l:dgBegin}--\ref{l:dgEnd}). The weighted edges of $dG$ are
then created from the edges of $G$: an edge $(u,v) \in E$ leads to two
edges $(u_{\mathrm{out}}, v_{\mathrm{in}})$ and $(v_{\mathrm{out}},
u_{\mathrm{in}})$, and edges that originate from a trusted node are
given capacity $2f{+}1$ and $1$ otherwise
(Alg.~\ref{alg:dolev_u_t_verif_2},
ll.~\ref{l:dgEdgeBegin}--\ref{l:dgEdgeEnd}).
Alg.~\ref{alg:dolev_u_t_verif_2} then studies all pairs of nodes $(u,v) \in V$. It first verifies whether $u$ and $v$ are neighbors in $G$ or if a path of trusted nodes connects them in $G$ (Alg.~\ref{alg:dolev_u_t_verif_2}, ll.~\ref{l:startconnected}--\ref{l:endconnected}). This verification requires distinguishing the various cases where $u$ and/or $v$ are trusted nodes, which uses functions $f$ and $g$ (cf. \S\ref{alg:dolev_u_t_verif_1}). Note that because the reliable communication abstraction assumes the broadcaster of a message to be correct, and therefore requires the weights of the outgoing edges of $u$ to be set to $2f+1$ (Alg.~\ref{alg:dolev_u_t_verif_2}, l.~\ref{line:modifyWeights}) and  later reset to their original values (Alg.~\ref{alg:dolev_u_t_verif_2}, l.~\ref{line:resetWeights}).  
The algorithm then verifies whether the maximum flow between $u$ and $v$ is at least $2f{+}1$ in $dG$. If all check fails, then the two nodes might not be able to communicate reliably in the presence of $f$ faulty nodes, and the algorithm returns \textsc{False} (Alg.~\ref{alg:dolev_u_t_verif_2}, l.~\ref{l:returnFalse}). If the checks pass for all pairs of nodes in $G$, then Alg.~\ref{alg:dolev_u_t_verif_2} concludes that \dut{} ensures \textit{RC-Validity} and therefore implements reliable communication in graph G. 

\noindent \textbf{Complexity.} 
App.~\ref{appx:method1}, computes the complexity of Alg.~\ref{alg:dolev_u_t_verif_2}, which is bounded by $\mathcal{O}(8{\cdot}f{\cdot}|V|^5)$.

\subsubsection{Method 2: Eliminating Unnecessary Trusted Nodes from \texorpdfstring{$G$}{G} and Checking for \texorpdfstring{$2f{+}1$}{2f+1} Vertex-Connectivity} 
\label{sec:dolev_u_t_verif_1}

Our second correctness verification method is more efficient when there are enough trusted nodes in the network. The method relies on Thms.~\ref{thm:untrusted_untrusted},~\ref{thm:trusted_untrusted} and~\ref{thm:trusted_trusted} and uses the aforementioned functions $f(\cdot)$ and $g(\cdot)$.

\begin{theorem}[\textit{RC-Validity} - untrusted  $\leftrightarrow$ untrusted]
\label{thm:untrusted_untrusted}
Two untrusted nodes $(u, v) \in (V_{\nauth} {\setminus} V_t)^2$ can
communicate reliably using \dut{} in $G$ iff. $u \in \Gamma_{f(G)}(v)$
or $\kappa_{f(G)}(u,v) \ge 2f{+}1$.
\end{theorem}

\begin{proof}
\emph{($\Rightarrow$)} If $u$ and $v$ can communicate reliably using \dut{} in $G$, then (i) they receive each other's messages with an empty path, which implies that they are neighbors in $f(G)$, or (ii) they are connected through $2f{+}1$ paths that become vertex-disjoint once their trusted nodes are removed and that connect them in $f(G)$. \\
\indent \emph{($\Leftarrow$ Part 1/2)} Let us assume that $u$ and $v$ are neighbors in $f(G)$. By definition, either they are also neighbors in $G$, or they are connected by a path of trusted nodes in $G$. In both cases, they can communicate reliably in $G$ using \dut{} because they obtain empty paths on each other's messages and deliver them. \\
\indent \emph{($\Leftarrow$ Part 2/2)} Let us now assume that $u$ and $v$ are connected through $2f{+}1$ vertex-disjoint paths in $f(G)$. Each of these paths is a subpath of at least a path that exists in $G$ and may contain additional intermediary trusted nodes. \dut{} removes trusted nodes from received paths, which means that $u$ and $v$ can always collect at least $f{+}1$ vertex-disjoint paths (that exist in $f(G)$) and authenticate each other's messages.  
\hfill $\OpenBox[1pt]$
\end{proof}
\vspace{2mm}

\begin{theorem}[\textit{RC-Validity} - trusted $\leftrightarrow$ untrusted] \label{thm:trusted_untrusted} A trusted node $u_t \in V_t$ and an untrusted node $v \in V_{\nauth}{\setminus}V_t$ can communicate reliably in $G$ iff. $u_t \in \Gamma_{g (G, f(G), u_t )}(v)$ or $\kappa_{g (G, f(G), u_t)}(u_t,v) \ge 2f{+}1$.
\end{theorem}

\begin{proof}
 $g \big(G, f(G), u_t \big)$ inserts trusted node $u_t$ in $f(G)$ and connects it with its untrusted neighbors in $G$ and with the untrusted nodes it can communicate with through a path of trusted nodes in $G$. Substituting $u$ by $u_t$, and $f(G)$ by $g \big(G, f(G), u_t \big)$ in the proof of Thm.~\ref{thm:untrusted_untrusted} proves this theorem. 
\hfill $\OpenBox[1pt]$
\end{proof}

\begin{theorem}[\textit{RC-Validity} - trusted  $\leftrightarrow$ trusted] \label{thm:trusted_trusted} 
Two trusted nodes $(u_{t},v_{t}) \in V_t^2$ can communicate reliably in $G$ iff. 
$u_t \in \Gamma_{g ( G, g (G, f(G), u_t ), v_t )}(v_t)$ or $\kappa_{g ( G, g (G, f(G), u_t ), v_t )}(u_t,v_t) \ge 2f{+}1$.
\end{theorem}

\begin{proof}
Once again the proof is immediate after substituting $u$ by $u_t$, $v$ by $v_t$ and $f(G)$ by $g \Big( G, g \big(G, f(G), u_t \big), v_t \Big)$ in the proof of Thm.~\ref{thm:untrusted_untrusted}. 

\hfill $\OpenBox[1pt]$
\end{proof}

\begin{algorithm}[ht]
\caption{Verification of \dut's \emph{RC-Validity} on a graph based on graph simplification and connectivity measurement (\dutverif{}, method 2).} 
\label{alg:dolev_u_t_verif_1}
\begin{algorithmic}[1]
\footnotesize 	
\State \textbf{Inputs:}
\State \hspace*{1em} $G=(V,E)$: undirected network topology.
\State \hspace*{1em} $V_t \subset V$: set of trusted nodes in $G$.
\State \hspace*{1em} $f$: max. number of Byzantine nodes.
\State \textbf{Output:}
\State \hspace*{1em} A boolean that indicates whether \dut is correct on $G$. 
\item[]

\State compute $f(G)$ \label{l:unionfind}

\State \textbf{forall} $(u,v) \in V \setminus V_t$ \textbf{and} $u < v$ \textbf{do}   \comalgo{Pairs of untrusted nodes} \label{l:untrusted_untrusted_begin}
    \State \hspace*{1em} \textbf{if} $\kappa_{f(G)}(u,v) < 2f{+}1$ \textbf{and} $u \notin \Gamma_{f(G)}(v)$ \textbf{then} \label{l:untrusted}
        \State \hspace*{2em} \textbf{return} $\textsc{False}$ \label{l:untrusted_untrusted_end}
\item[]

\State \textbf{forall} $u_t \in V_t$ \textbf{do} 
    \State \hspace*{1em} compute $G_{u_t} = g(G,f(G),u_t)$ \label{l:computeg}
    
    \State \hspace*{1em} \textbf{forall} $v \in V \setminus V_t$ \textbf{do} \comalgo{Pairs of trusted-untrusted nodes} \label{l:trusted_untrusted_begin}
        \State \hspace*{2em} \textbf{if} $\kappa_{G_{u_t}}(u_t,v) < 2f{+}1$ \textbf{and} $u_t \notin \Gamma_{G_{u_t}}(v)$ \textbf{then} \label{l:connecutv}
            \State \hspace*{3em} \textbf{return} $\textsc{False}$ \label{l:trusted_untrusted_end}
    
    \State \hspace*{1em} \textbf{forall} $v_t \in V_t \setminus \{ u_t \}$ s.t. $u_t < v_t$ \textbf{do} \comalgo{Pairs of trusted nodes} \label{l:trusted_trusted_begin}
        \State \hspace*{2em} compute $G_{u_t,v_t} = g(G,G_{u_t},v_t)$ \label{XXX} \label{l:gutvt}
        \State \hspace*{2em} \textbf{if} $\kappa_{G_{u_t,v_t}}(u_t,v_t) < 2f{+}1$ \textbf{and} $u_t \notin \Gamma_{G_{u_t, v_t}}(v_t)$ \textbf{then} \label{l:connecutvt}
            \State \hspace*{3em} \textbf{return} $\textsc{False}$ \label{l:trusted_trusted_end}

\item[]
\State \textbf{return} $\textsc{True}$
\end{algorithmic}
\end{algorithm}

\noindent \textbf{RC-Validity Verification Algorithm.}
Alg.~\ref{alg:dolev_u_t_verif_1} 
directly leverages Thms.~\ref{thm:untrusted_untrusted},~\ref{thm:trusted_untrusted}, and~\ref{thm:trusted_trusted} to verify whether \dut{} ensures reliable communication between any pairs of nodes on a given graph $G$, assuming a set $V_t$ of nodes are trusted and knowledge of the value $f$.
The algorithm first verifies whether any two untrusted nodes can communicate reliably, which is the case if they are at least $2f{+}1$ vertex-connected or neighbors in $f(G)$ (Alg.~\ref{alg:dolev_u_t_verif_1}, ll.~\ref{l:untrusted_untrusted_begin}--\ref{l:untrusted_untrusted_end}). If so, Thm.~\ref{thm:untrusted_untrusted} ensures that it will also be the case in $G$. Next, the algorithm similarly verifies whether each trusted node $u_t$ can communicate reliably with all untrusted nodes, i.e., if they are neighbors or if they are $2f{+}1$ vertex-connected (Alg.~\ref{alg:dolev_u_t_verif_1}, ll.~\ref{l:trusted_untrusted_begin}--\ref{l:trusted_untrusted_end}) in $G_{u_t} = g(G,f(G),u_t)$. If so, Thm.~\ref{thm:trusted_untrusted} ensures that they can also communicate reliably in $G$. Finally, the algorithm verifies that any two trusted nodes $u_t$ and $v_t$ can communicate reliably in $G_{u_t, v_t} = g(G,G_{u_t},v_t)$ (Alg.~\ref{alg:dolev_u_t_verif_1}, ll.~\ref{l:trusted_trusted_begin}--\ref{l:trusted_trusted_end}). Thm.~\ref{thm:trusted_trusted} then ensures that $u_t$ and $v_t$ can communicate reliably in $G$. If $\textsc{False}$ was not returned by the routine after executing all checks, the algorithm then concludes that reliable communication is enforced by \dut on graph $G$ and returns $\textsc{True}$. 

\noindent \textbf{Complexity.} 
App.~\ref{appx:method2} computes the complexity of Alg.~\ref{alg:dolev_u_t_verif_1}. When $|V_t|$ is small, $|V {\setminus} V_t|^5 \approx |V|^5$ dominates the complexity. When $|V_t|$ is sufficiently large, then the complexity is dominated by $|V_t|^2 {\cdot} |V| \approx |V|^3$. 

\section{\sft: Reliable Communication in the Authenticated Process Model with Trusted Nodes}
\label{sec:sigflood_t}

As a stepping stone towards our final protocol, we now consider the
system model where all nodes can use digital signatures (i.e.,
$V_{\auth} {=} V$ and $V_{\nauth} {=} \emptyset$). The state-of-the-art
reliable communication protocol in unknown networks with authenticated
processes uses network flooding to disseminate a payload signed by the
original broadcaster. In this protocol, the sender of a broadcast
signs its original message and sends it to all its neighbors. Upon
receiving a payload and a valid signature, nodes verify its
authenticity and forward both to all their neighbors that have not yet
transmitted the message~\cite{KermarrecS07}.

\subsection{Description}

We present the pseudocode of \sft, our reliable communication algorithm for the authenticated process model that leverages trusted nodes in unknown topologies, in Alg.~\ref{alg:floodingtrusted}. 
Note that we do not detail the handling of incorrect signatures for simplicity. 
In \sft, messages contain the payload data $m$ that is being broadcast, the identity of the sender $p_s$ and its signature $\sigma_s(m)$ (Alg.~\ref{alg:floodingtrusted}, l.~\ref{line:msgType2}).
A node that receives a message from a trusted neighbor immediately delivers it, as it considers that its trusted neighbor authenticated it.
Trusted components are not leveraged by \sft, since they do not have direct network access, which implies that their signature still needs to be verified. 

In Alg.~\ref{alg:floodingtrusted}, the broadcasting node initially sends its message along with its signature to all its neighbors (Alg.~\ref{alg:floodingtrusted}, l.~\ref{ft:begin}--\ref{ft:end}). 
When a node receives payload data along with a signature and the broadcaster’s identity, it either verifies the correctness of the signature or, if the message is from a trusted neighbor, accepts it immediately (Alg.~\ref{alg:floodingtrusted}, l.~\ref{ft:verif}). 
Upon receiving and successfully verifying a message for the first time, a node forwards the message, along with its signature and the sender's ID, to all its neighbors except the one from which the message was originally received.

Termination of \sft is guaranteed since each correct node forwards a given message at most once. \sft generates at most two messages per network link, which occurs when two neighbors simultaneously send a message to each other. In practice, it is therefore expected that \sft would generate significantly fewer messages than \dut on a given network topology, whose base version generates $O(N!)$ messages where $N$ is the system size.

\begin{algorithm}
\caption{\sft: Reliable communication in $(f{+}1)$-connected networks at process~$p_i$ in the presence of trusted processes in the authenticated process model. } \label{alg:floodingtrusted}
\begin{algorithmic}[1]
\footnotesize
\State \textbf{Parameters:}
\State \hspace*{1em} $f$ : max. number of Byzantine processes in the system.
\State \textbf{Uses:}
    \State \hspace*{1em} $\bullet$ Auth. async. perfect point-to-point links, instance \textit{al}.
    \State \hspace*{1em} $\bullet$ $G=(V,E)$: network topology
    \State \hspace*{1em} $\bullet$ $\sigma_i(m)$ is the signature of process $p_i$ over message $m$, and $\texttt{verif}(\sigma_j(m))$ verifies it. 
\item[]

\State \underline{\textbf{upon event} $\langle \sft{},\ \texttt{Init} \rangle$ \textbf{do}}
    \State \hspace*{1em} $delivered = \textbf{False}$
\item[]

\State \underline{\textbf{upon event} $\langle \sft{},\ \texttt{Broadcast}\ |\ m \rangle$ \textbf{do}} \label{ft:begin}
    \State \hspace*{1em} \textbf{forall} $p_j \in \Gamma_G(p_i)$ \textbf{do}
        \State \hspace*{2em} $\langle al, \texttt{Send}\ |\ p_j, [m, \sigma_i(m)] \rangle$
    \State \hspace*{1em} $delivered = \textbf{True}$
    \State \hspace*{1em} $\langle \sft{}, \texttt{Deliver}\ |\ m \rangle$ \label{ft:end}
\item[]

\State \underline{\textbf{upon event} $\langle al, \texttt{Receive}\ |\ p_j, [m, \sigma_s(m)] \rangle$ \textbf{do}} \label{line:msgType2}
    \State \hspace*{1em} \textbf{if} $delivered$ \textbf{then} \textbf{return}
    \State \hspace*{1em} \textbf{if} \textcolor{blue}{$\texttt{isTrusted}(p_j)$} \textbf{or} $\texttt{verif}(\sigma_s(m))$ \textbf{then} \label{ft:verif}
        \State \hspace*{2em} $delivered = \textbf{True}$
        \State \hspace*{2em} $\langle \sft{}, \texttt{Deliver}\ |\ m \rangle$ 
        \State \hspace*{2em} \textbf{forall} $p_k \in \Gamma_G(p_i) \setminus \{p_j,p_s\}$ \textbf{do}
            \State \hspace*{3em} $\langle al, \texttt{Send}\ |\ p_k, [m, \sigma_s(m)] \rangle$
\end{algorithmic}
\end{algorithm}

\subsection{Proofs}
\label{sec:sigflood_t_verif}

Lemma~\ref{lemma:8} shows that \sft{} implements the reliable communication properties. 
  
  \begin{lemma} \label{lemma:8}
\sft{} maintains \textbf{RC-no duplication} and \textbf{RC-no
  creation}, and \textbf{RC-Validity}.
\end{lemma}

\begin{proof}
\textit{RC-no duplication} is guaranteed by the use of the $delivered$ variable. \\

A message is delivered only when the signature of its creator is verified, which ensures \textit{RC-no creation.}  \\

Finally, \textit{RC-Validity} can be verified on a given network thanks to Alg.~\ref{alg:dolev_u_t_verif_1} or Alg.~\ref{alg:dolev_u_t_verif_2}. 
In the absence of trusted nodes, \sft{} requires the network to be at
least $(f{+}1)$ connected to ensure that the broadcaster's signature reaches every node. Conversely, \dut requires the network to be at least $(2f{+}1)$ connected. Verifying the correctness of \sft{} on a graph can be done using Alg.~\ref{alg:dolev_u_t_verif_1} or Alg.~\ref{alg:dolev_u_t_verif_2} by replacing all instances of $2f{+}1$ with $f{+}1$ in those algorithms.

\hfill $\OpenBox[1pt]$
\end{proof}
\vspace{1mm}

\section{\dualrc: Reliable Communication in the Hybrid Authenticated Links/Processes Model with trusted nodes and components}
\label{sec:dualrc}

We now consider the most generic reliable communication scenario that involves both non-authenticated and authenticated nodes. Our final algorithm, \dualrc,  builds on our two previous protocols, \dut{} and \sft{}, and therefore inherits their ability to leverage trusted nodes.  However, contrary to these two protocols, \dualrc{} also uses specific optimizations that require manipulating both dissemination paths and signatures. In particular, \dualrc{} leverages the presence of trusted components by allowing them to transform a set of signed and disjoint dissemination paths into a single trusted signature. We give the pseudocode of \dualrc{} in Alg.~\ref{alg:dualrc}, where the blue text is specific to the presence of trusted components. Once again, we omit the handling of incorrect signatures for simplicity.

\begin{algorithm*}
\caption{\dualrc: Reliable communication at process $p_i$ with trusted and TC-augmented processes in the hybrid link/process authenticated model. Modifications MD.1--5 are applied to messages that carry dissemination paths (cf. \S\ref{sec:bonomi_modifications}). Text that appears in blue is specific to the presence of trusted components.} 
\label{alg:dualrc}

\setlength{\columnsep}{-2.4cm} 
\begin{multicols}{2}

\begin{algorithmic}[1] 
\footnotesize 	

\State \textbf{Parameters:}
\State \hspace*{1em} $f$ : max. number of Byzantine nodes in the system.
\State \textbf{Uses:} 
\State \hspace*{1em} $\bullet$ Auth. async. perfect point-to-point links, instance \textit{al}.
\State \hspace*{1em} $\bullet$ $G=(V,E)$: network topology 
\State \hspace*{1em} $\bullet$ $\Gamma_G(p_j)$ returns the list of $p_j$'s neighbors. 
\State \hspace*{1em} $\bullet$ $\texttt{isTrusted}(p_j)$ indicates if $p_j$ is trusted.  
\State \hspace*{1em} $\bullet$ \textcolor{blue}{$\texttt{hasTC}(p_j)$ indicates if $p_j$ has a trusted component.}  
\State \hspace*{1em} $\bullet$ $\texttt{isAuth}(p_j)$ indicates if $p_j$ is authenticated. 
\State \hspace*{1em} $\bullet$ \textcolor{blue}{$\texttt{isTC}(\sigma)$ indicates if a signature was generated by a TC.} 
\State \hspace*{1em} $\bullet$ (\underline{Auth. nodes only}) $\sigma_i(m)$ is the signature of process $p_i$ \\ over message $m$, and $\texttt{verif}(\sigma_j(m))$ verifies it.
\State \hspace*{1em} $\bullet$ \textcolor{blue}{(\underline{Nodes equipped with a TC})  $\sigma_{{TC}_i}(m)$ is the signature \\of $p_i$'s TC.}
\item[]

\State \underline{\textbf{upon event}{$\langle \dualrc{}, \texttt{Init} \rangle$} \textbf{do}} \label{l:dualrcinit}
    \State $delivered = \textbf{False}$
    \State $paths = \emptyset$ \comalgo{Received dissemination paths}
    \State \textcolor{blue}{$sPaths = \emptyset$} 
    \comalgo{Received signed dissemination paths} 
    \State $dN = \emptyset$ \comalgo{Set of nodes that delivered (for MD.3)}
    \State $sigs = \emptyset$ \comalgo{Set of received signatures}
\item[]

\State \underline{\textbf{upon event} $\langle \dualrc{}, \texttt{Broadcast}\ |\ m \rangle$ \textbf{do}}
    \State \textbf{forall} $p_j \in \Gamma_G(p_i)$ \textbf{do}
        \State \hspace*{1em} \textbf{if} $\texttt{isAuth}(p_i)$ \textbf{then}
            \State \hspace*{2em} $\langle al, \texttt{Send}\ |\ p_j, [m, i, \sigma_i(m, i) ] \rangle$ \label{line:sendAuthBcast}
        \State \hspace*{1em} $\langle al, \texttt{Send}\ {|}\ p_j, [ m, i, \emptypath, \textcolor{blue}{\emptylist} ] \rangle$  \label{line:sendNAuthBcast}
    \State $\texttt{deliver}(m)$

\item[]
\State \underline{\textbf{function} $\texttt{deliver}(m)$}
\State \textbf{if} \textbf{not} $delivered$ \textbf{then}
\State \hspace*{1em} $delivered = \textbf{True}$
\State \hspace*{1em} $\langle \dualrc{}, \texttt{Deliver}\ |\ m \rangle$

\item[]
\State \comalgo{Receive a single signature}
\State \underline{\textbf{upon event} $\langle al, \texttt{Receive}\ |\ p_j, [ m, s, \sigma_l(m,s)] \rangle$ \textbf{do}} \label{l:msgType1}

    \State \textbf{if} $\sigma_l(m,s) \in sigs$ \textbf{then}  \textbf{return;}
     \textbf{else} $sigs {=} sigs {\cup} \{\sigma_l(m,s)\}$
    
    \State \textbf{if} $p_j \eqeq p_l$ \textbf{or} (\texttt{isTrusted}($p_j$) \textbf{and} \texttt{isAuth}($p_j$)) \textbf{then} \label{l:directDeliver}
        \State \hspace*{1em} $dN = dN \cup \{p_l\}$ 
        \State \hspace*{1em} $\texttt{deliver}(m)$

   \item[]  
   \State $dests = \Gamma_G(p_i) \setminus \{ p_s, p_l \}$
    \State \textbf{if} $isAuth(p_i)$ \textbf{then} \comalgo{try to deliver and broadcast own sig.}
        \State \hspace*{1em} $sPaths.\texttt{insert}(\{p_l\} \setminus \{p_s\})$ \label{line:singleSignToPath}
        \State \hspace*{1em} $sPaths\_cond {=} \exists \{f{+}1$ vertex-disjoint paths\}${\subset} sPaths$
        \State \hspace*{1em} \textbf{if} ($p_s \eqeq p_l$ \textbf{or} $\texttt{isTrusted}(p_l)$ \textcolor{blue}{\textbf{or}} \\ \hspace*{1em} \textcolor{blue}{$\texttt{isTC}(\sigma_l(m))$} \textbf{or} $sPaths\_cond$) \textbf{then} \label{line:deliverSign}
            \State \hspace*{2em} $\texttt{deliver}(m)$
            \State \hspace*{2em} \textbf{if} $\texttt{isTrusted}(p_l)$ \textcolor{blue}{\textbf{or}  $\texttt{isTC}(\sigma_l(m, s))$} \textbf{then}
                \State \hspace*{3em} \textbf{forall} $p_k \in dests$    \textbf{do} 
                    \State \hspace*{4em} $\langle al, \texttt{Send}\ |\ p_k, [m, s,\sigma_l(m,s) ] \rangle$ \comalgo{fwd sender's sig.} \label{l:fwdTrusted}
                \State \hspace*{3em} \textbf{return}

            \State \textcolor{blue}{\hspace*{2em} \textbf{else if} $\texttt{hasTC}(p_i)$ \textbf{then}} 
                \State \hspace*{3em} \textcolor{blue}{\textbf{forall} $p_k \in dests$ \textbf{do}} 
                    \State \hspace*{4em} \textcolor{blue}{$\langle al, \texttt{Send}\ |\ p_k, [m, s, \sigma_{TC_i}(m,s) ] \rangle$} \comalgo{send TC's sig.} \label{line:tcsign} \label{l:fwdTC}
                \State \hspace*{3em} $sigs = sigs \cup \{ \sigma_{TC_i}(m,s) \}$
                \State \hspace*{3em} \textbf{return}

            \State \hspace*{2em} \textbf{else} 
                \State \hspace*{3em} \textbf{forall} $p_k \in dests$ \textbf{do} 
                    \State \hspace*{4em} $\langle al, \texttt{Send}\ |\ p_k, [m, s,\sigma_i(m,s) ] \rangle$ \comalgo{send own sig.} \label{l:fwdSign}

                \State \hspace*{3em} $sigs = sigs \cup \{\sigma_i(m,s)\}$ 

   \item[]  
   \State \textbf{forall} $p_k \in dests$ \textbf{do}
       \State \hspace*{1em} $\langle al, \texttt{Send}\ |\ p_k, [m, s, \sigma_l(m,s)] \rangle$ \comalgo{forward (non-trusted) sig.} \label{l:finalFwd}

\columnbreak 

\State \comalgo{Receive a single path, with signatures on subpaths}
\State \underline{\textbf{upon event} $\langle al, \texttt{Receive}\ |\ p_j, [ m,s, path\textcolor{blue}{, List(pa, \sigma(pa, m, s))}] \rangle$ \textbf{do}} \label{l:msgType2} 
    \State \textbf{if} $delivered$ \textbf{or} $path \cap dN \neq \emptyset$ \textbf{then} \textbf{return} \comalgo{MD.5 and MD.4}
    \State \textbf{if} $path \eqeq \emptypath$ \textbf{then} $dN = dN \cup \{p_j\}$ \comalgo{Check if dN is correct - MD.3}
    
    \item[]
    \State \comalgo{Try to identify $f{+}1$ 
    signed vertex-disjoint paths.} 
    \State \textcolor{blue}{\textbf{if} $\texttt{isAuth}(p_i)$ \textbf{then}}
    \State \hspace*{1em} \textcolor{blue}{\textbf{forall} $pa, \sigma_l(pa, m, s) \in List(pa, \sigma(pa, m, s))$ \textbf{do}} \label{l:beginSignedPaths}
        \State \hspace*{2em} \textcolor{blue}{$sPaths.$\texttt{insert}(\texttt{rm\_all\_trusted\_TCs}$(pa{\cup}\{p_l\}))$} \label{line:signToPath3}
        \State \hspace*{2em} \textcolor{blue}{\textbf{if} $\emptypath \in sPaths$ \textbf{or} $\exists \{f{+}1$ vertex-disjoint paths$\} {\subset} sPaths$ \textbf{then}}
            \State \hspace*{3em} \textcolor{blue}{$\texttt{deliver}(m)$}
            \State \hspace*{3em} \textcolor{blue}{\textbf{forall} $p_k \in  \Gamma_G(p_i)$ \textbf{do}} 
                \State \hspace*{4em} \textcolor{blue}{\textbf{if} $\texttt{hasTC}(p_i)$ \textbf{then}}
                    \State \hspace*{5em} \textcolor{blue}{$\langle al, \texttt{Send}\ |\ p_k, [m, s, \sigma_{{TC}_i}(m,s) ] \rangle$}  \label{line:teeSignedPaths} 
                \State \hspace*{4em} \textcolor{blue}{\textbf{else}
                    \State \hspace*{5em} $\langle al, \texttt{Send}\ |\ p_k, [m, s, \sigma_i(m, s) ] \rangle$} 
            \State \hspace*{3em} \textcolor{blue}{\textbf{forall} $p_k \in \Gamma_G(p_i) {\setminus} dN$ \textbf{do}} \comalgo{MD.2 + MD.3}
                \State \hspace*{4em} \textcolor{blue}{$\langle al, \texttt{Send} | p_k, [m, s, \emptypath, \sigma_i(\emptypath, m, s) {+} List(pa, \sigma(pa, m, s)) ] \rangle$} 
           \State \hspace*{3em} \textcolor{blue}{\textbf{return}} \label{l:endSignedPaths}

\item[]
    \State \comalgo{$p_i$ is not authenticated or signed paths did not cause a deliver} 
    \State $rpath = path \cup \{ p_j \}$ \comalgo{path to forward} \label{l:dutpathb}
    \State $path = \texttt{remove\_all\_trusted}(rpath)$ \label{line:rmTrusted3}
    \State $paths.\texttt{insert}(path)$ 
    
\item[]
    \State \comalgo{try to deliver using dissemination paths} \label{line:delPaths3}
    \State \textbf{if} $\emptypath \in paths$ (MD.1) \textbf{or} $\exists \{f{+}1$ vertex-disjoint paths$\} {\subset} paths$ \textbf{then}
        \State  \hspace*{1em} \textbf{if} $\texttt{isAuth}(p_i)$ \textbf{then} %
            \State \hspace*{2em} \textbf{forall} $p_k \in \Gamma_G(p_i)$ \textbf{do} \comalgo{MD.2-3}
                \State \hspace*{3em} $\langle al, \texttt{Send}\ |\ p_k, [m, s, \sigma_i(m, s)] \rangle$
            \State \hspace*{2em} \textbf{forall} $p_k \in \Gamma_G(p_i) {\setminus} dN$ \textbf{do} \comalgo{MD.2-3}
                \State \hspace*{3em} $\langle al, \texttt{Send}\ |\ p_k, [m, s, \emptypath\textcolor{blue}{, \sigma_i(\emptypath, m, s) {+} List(pa, \sigma(pa, m, s))}] \rangle$ 
        \State \hspace*{1em} \textbf{else}
            \State \hspace*{2em} \textbf{forall} $p_k \in \Gamma_G(p_i) {\setminus} dN$ \textbf{do} \comalgo{MD.2-3}
                \State \hspace*{3em} $\langle al, \texttt{Send}\ |\ p_k, [m, s, \emptypath\textcolor{blue}{, List(pa, \sigma(pa, m, s))} ] \rangle$
                
        \State \hspace*{1em} $\texttt{deliver}(m)$
        \State \hspace*{1em} \textbf{return} \label{l:dutpathe}
        
    \item[]
    \State  \comalgo{If not delivered, forward message with modified path}
    \State  \textbf{forall} $p_l \in \Gamma_G(p_i) \setminus rpath$ \textbf{do} \label{line:finalRelay3b} 
        \State \hspace*{1em} \textbf{if} $\texttt{isAuth}(p_i)$ \textbf{then}
            \State \hspace*{2em} $\langle al, \texttt{Send}\ |\ p_l, [m, s, rpath\textcolor{blue}{, \sigma_i(rpath, m, s)} \textcolor{blue}{+ List(pa, \sigma(pa, m, s))} ]\rangle$ 
        \State \hspace*{1em} \textbf{else}
            \State \hspace*{2em} $\langle al, \texttt{Send}\ |\ p_l, [m, s, rpath\textcolor{blue}{, List(pa, \sigma(pa, m, s))} ]\rangle$ \label{line:finalRelay3e}

\end{algorithmic}

\end{multicols}

\end{algorithm*}

\subsection{Description}

\textbf{Message types.} \dualrc{} relies on both dissemination paths and digital signatures to allow all processes to authenticate a message. We describe in the following the two types of messages that \dualrc manipulates. \\
$\bullet$ A \emph{signature message} $[m, s, \sigma_l(m, s)]$ contains a payload $m$, the id $s$ of the original sender $p_s$ and a signature generated by process $p_l$ over $(m, s)$. All processes can manipulate and interpret signature messages to deliver messages.  \\
$\bullet$ A \emph{path message} $[m, s, path, List(pa, \sigma(pa, m, s))]$ contains a payload $m$, the id $s$ of the original sender $p_s$, an unsigned dissemination path $path$, and a possibly empty list of signed subpaths $List(pa, \sigma(pa, m, s))$. Path messages can only be interpreted by authenticated processes, but non-authenticated processes might have to forward them anyway. Modifications MD.1-MD.5 (cf. \S\ref{sec:bonomi_modifications}) are used when manipulating path messages. Note that, as an optimization, it would be possible to compress a list of signed paths to a single path accompanied by an onion signature. 

\textbf{Initialization.} \dualrc{} relies on variables that are set upon a node's initialization: The $delivered$ boolean indicates whether the message has been delivered (init. set to $\textsc{False})$; $paths$ and $signedPaths$ respectively contain the unsigned and the signed dissemination paths that have been received (init. set to $\emptyset$); $dN$ is the set of nodes that are known to have delivered the message (init. set to $\emptyset$); $sigs$ contains the signatures that have been received (init. set to $\emptyset$). 

\textbf{Broadcasting a message.} When a node $p_i$ initiates a broadcast, it sends a path message $[m,i,\emptypath, \emptylist]$ to all its neighbors (Alg.~\ref{alg:dualrc}, l.~\ref{line:sendNAuthBcast}). This message carries an empty dissemination path $\emptypath$ and an empty list $\emptylist$ of signed subpaths. If node $p_i$ is authenticated then it additionally sends a signature message  $[m, j, \sigma_i(m, i)]$ to its neighbors (Alg.~\ref{alg:dualrc}, l.~\ref{line:sendAuthBcast}). Node $p_i$ immediately delivers a message it broadcasts. 

Signature and path messages are handled by two distinct procedures that all types of nodes execute, and which we describe in the following. 

\textbf{Receiving a signature message $[m, s, \sigma_l(m, s)]$ (Alg.~\ref{alg:dualrc}, l.~\ref{l:msgType1}).} In \dualrc{}, a node generates its signature on a message only after it delivers it. Received signatures are stored in $sigs$, and are processed exactly once. A message can be delivered without verifying its signature if it is received directly from the broadcaster, or if it is relayed by a trusted and authenticated node (Alg.~\ref{alg:dualrc}, l.~\ref{l:directDeliver}). Authenticated nodes can also deliver $m$ using signature $\sigma_l(m,s)$ if it was generated by the broadcaster ($p_s == p_l$), by a trusted node or by component, or if they identify $f{+}1$ disjoint signed paths in $sPaths$ (a signature is interpreted as a single-node signed path, except for the broadcaster's signature that is interpreted as an empty path) (Alg.~\ref{alg:dualrc}, l.~\ref{line:deliverSign}). If so, a node then sends a signature to the subset $dests = \Gamma_G(p_i) \setminus \{ p_s, p_l \}$ of its neighbors, which might not have delivered $m$. If $\sigma_l(m,s)$ is a trusted signature then $p_i$ forwards it to the nodes in $dests$ and returns (Alg.~\ref{alg:dualrc}, l.~\ref{l:fwdTrusted}). Otherwise, if it is equipped with a trusted component then it makes it verify $\sigma_l(m,s)$ and generate a trusted signature $\sigma_{{TC}_i}(m,s)$ that it then sends to the nodes in $dests$ and returns (Alg.~\ref{alg:dualrc}, l.~\ref{l:fwdTC}). If it is authenticated, and if it has not forwarded or generated a trusted node's or component's signature, then node $p_i$ sends its own signature to the nodes in $dests$ (Alg.~\ref{alg:dualrc}, l.~\ref{l:fwdSign}). Finally, a node that has not returned forwards the signature it received to the processes in $dests$ (Alg.~\ref{alg:dualrc}, l.~\ref{l:finalFwd}). 

\textbf{Receiving a path message $[m, s, path,$ $List(pa, \sigma(pa, m, s))]$ (Alg.~\ref{alg:dualrc}, l.~\ref{l:msgType2}).}  A node might ignore a path message using MD.4 or MD.5 as in \dut{}. If $path == \emptypath$ then node $p_j$ is added to $dN$, which tracks the nodes that delivered a path message

If node $p_i$ is authenticated, then it tries to leverage each path $pa$ and signed path $\sigma_l(pa, m, s)$ contained in $List(pa, \sigma(pa, m, s))$ (Alg.~\ref{alg:dualrc}, ll.~\ref{l:beginSignedPaths}--\ref{l:endSignedPaths}). First, the path $pa \cup \{p\}]$ to which trusted signatures are removed is added to $sPaths$. If $sPaths$ contains an empty path (i.e., the broadcaster's signature was received) or $f{+}1$ vertex-disjoint paths then node $p_i$ delivers $m$. Node $p_i$ then broadcasts its trusted component's signature, if it is equipped with one, and otherwise broadcasts its own signature in a signature message. More precisely, node $p_i$ transfers the $f{+}1$ signed vertex-disjoint paths it identified to its trusted component so that it can verify their signature, check that they are disjoint, and if so return a signature $\sigma_{{TC}_i}(m,s)$ (Alg.~\ref{alg:dualrc}, ll.~\ref{line:tcsign} and~\ref{line:teeSignedPaths}). 
Node $i$ then also broadcasts a path message with an empty path to all its neighbors in $\Gamma_G(p_i) \setminus dN$ and returns. 

If node $p_i$ has not delivered $m$ using signed paths, it attempts to do so with $path$, as in \dut{} (Alg.~\ref{alg:dualrc}, ll.~\ref{l:dutpathb}--\ref{l:dutpathe}). This path is the received path to which $p_j$ has been added and from which all trusted nodes have been removed (Alg.~\ref{alg:dualrc}, l.~\ref{line:rmTrusted3}). Note that processes equipped with a trusted component cannot be removed from dissemination paths because unsigned dissemination paths are not manipulated by trusted components. If node $p_i$ delivers $m$ based on a dissemination path, it  then broadcasts its signature to all its neighbors in a signature message and broadcasts a path message that contains an empty path and the list of signed path it received to which an appended signed path has been added to the nodes in $\Gamma_G(p_i) \setminus dN$. If node $p_i$ is not authenticated then it only broadcasts a path message with an empty path and the received list of signed paths. A node that delivers using $path$ then returns. 

Finally, if node $p_i$ has not delivered payload $m$, it forwards the message it received with up to two modifications: it sets path $rpath = path \cup \{p_j\}$, and if it is authenticated, it appends $\sigma_i(rpath, m, s)$ to the list of signed paths (Alg.~\ref{alg:dualrc}, ll.~\ref{line:finalRelay3b}--\ref{line:finalRelay3e}). 

\textbf{Termination and Optimizations.} A node stops processing and forwarding path messages as soon as it delivers (as in \dut{}). 
However, a node still has to forward signatures after delivering to allow all authenticated nodes to deliver (App.~\ref{appx:counter_example} shows an example where this is necessary). Termination is ensured because a node processes and relays a signature at most once through the use of variable $sigs$. 

Limiting the dissemination of signatures further, in particular by non-authenticated nodes, is more challenging and our optimizations had to be left out of  Alg.~\ref{alg:dualrc} for space reasons. First, a signature should not be relayed to a neighbor that sent the broadcaster's signature, a trusted signature, or $f{+}1$ untrusted signatures. Additionally, authenticated nodes can stop forwarding signature messages after they have forwarded the broadcaster's signature, a trusted signature, or $f{+}1$ untrusted signatures, as transmitting more signatures would not carry more information (the message has been authenticated). Non-authenticated nodes can stop forwarding signatures after they have relayed $f+1$ different signatures that they can reliably trace back to $f{+}1$ different authenticated nodes that relayed them or emitted them. Indeed, each of these untrusted signatures has been verified by $f{+}1$ different nodes and at least one of them can correctly be assumed to have successfully and correctly verified the signatures. 

Another optimization consists in encoding the list of signed paths a message path contains as a single onion signature or as an aggregate signature. This would effectively reduce the bit complexity of \dualrc{}. As a dissemination path should not be transferred to a neighbor that delivered (MD.3), an (authenticated or unauthenticated) node can avoid sending path messages to a node that transmitted the broadcaster's signature, a trusted signature, its own signature or $f{+}1$ untrusted signatures.   
Finally, it might be possible to avoid creating unsigned path messages between authenticated nodes, but we believe that it would render \dualrc{}'s code significantly more complex.  

\subsection{Proofs}

\noindent \textbf{Delivery requirements.} 
\dualrc{}'s pseudocode (Alg.~\ref{alg:dualrc}) allows a node to deliver a message under two conditions (DR.1 \& DR.2):

\noindent [DR.1]: A non-authenticated node $v$ delivers a message from a broadcaster $u$ if at least one of subconditions NA.1 and NA.2 is met.
    \begin{description}
        \item[NA.1] Node $v$ computes an empty path after removing the trusted nodes it may contain. 
        \item[NA.2] Node $v$ computes $f{+}1$ paths that are vertex-disjoint after removing the trusted nodes they contain.
        \label{desc}
    \end{description}

\noindent [DR.2] An authenticated node $v$ delivers a message from a broadcaster $u$ if subconditions NA.1 or NA.2 are met or if at least one of subconditions A.1 or A.2 is met. 
    \begin{description}    
        \item[A.1] Node $v$ receives node $u$'s, a trusted node's or a trusted component's signature.
        \item[A.2] Node $v$ computes a combination of $f{+}1$ signatures or paths that are vertex-disjoint after removing the trusted nodes from the paths.
    \end{description}

\noindent Note that A.2 generalizes NA.2 for authenticated nodes. Unfortunately, conditions NA.1, NA.2. A.1, and A.2 do not directly lead to a correctness verification algorithm one could apply on a given graph. To reach this objective, we formulate equivalent delivery conditions in Thm.~\ref{thm:5}. These conditions assume a set $S$ of nodes known to have delivered a message broadcast by a process $u$ and require manipulation to the directed graph $dG$, generated as described in \S\ref{sec:dolev_u_t_verif_2}, so that max-flows can be computed. As a brief summary, graph $dG$ is obtained from $g$ using the following steps: (i) nodes are split into two, and edges are added between them with a capacity of $1$ if the original node is not trusted or $2f{+}1$ if the original node is trusted, and (ii) edges of the original graph $G$ are transformed into two directed edges in $dG$ of weight $1$ if their origin is not trusted or $2f{+}1$ if their origin is trusted. 

We introduce synthetic \emph{root} nodes that are
connected to different subsets of nodes in $S$ using weighted directed
edges (from a root to the nodes).
\begin{itemize}
\item Node $root_S$ is connected to all nodes in $S \cap V$ using edges of capacity $1$ (for untrusted nodes) or $2f{+}1$ for trusted nodes.
\item Node $root_{n.auth}$ is connected to all (non-authenticated) nodes in $S \cap V_{n.auth}$ with edges of weight $1$. 
\end{itemize}

\begin{theorem}\label{thm:5}
A node $v$ is guaranteed to always deliver a message that is broadcast by a (correct) node $u$ iff. at least one of the following conditions is verified. 
    \begin{description}
        \item[5.1] Nodes $u$ and $v$ are neighbors in $G$. 
        \item[5.2] The max flow from $root_S$ to $v$ in $dG$ is at least $2f{+}1$.
        
        \item[5.3] Nodes $u$ and $v$ are authenticated, and the max flow from $u$ to $v$ in $dG$ is larger than $f{+}1$.
        \item[5.4] Node $v$ is authenticated and there is at least one trusted and authenticated node in $S$ whose max flow in $dG$ to $v$ is greater than or equal to $f{+}1$. 
    \end{description}
\end{theorem}

\begin{proof}

\emph{($\Rightarrow$)}
Condition NA.1 can be met in two cases: (i) $u$ and $v$ are neighbors (condition 5.1), or (ii) they are connected by a path of trusted nodes (a subcase of condition 5.2).

Condition NA.2 can occur following several scenarios. First, if there are $2f{+}1$ vertex-disjoint paths from $u$ to $v$. Second, if at least $2f{+}1$ untrusted nodes deliver a message and, therefore, send an empty path to their neighbors as a result of which eventually $v$ always receives $f{+}1$ disjoint paths. Third, if $v$ is able to receive at least $2f{+}1$ non-empty dissemination paths that originate from trusted nodes. Finally, if $v$ should be able to receive a combination of $2f{+}1$ vertex-disjoint paths and signatures, that are also used as single-node paths. In each case, the flow from $root_S$ to $v$ in $dG$ is at least $2f+1$ (condition 5.2), which implies that $f{+}1$ nodes in $S$ (possibly including $u$) are correct and will broadcast empty paths.
Node $v$ will then always receive $f{+}1$ disjoint paths as a consequence of these delivering nodes. As dissemination paths are not relayed by
nodes that deliver a message, one could wonder whether $v$ will always receive these $f{+}1$ disjoint paths. This is indeed the case because a
delivering node located on one of the $f{+}1$ correct paths would
still broadcasts an empty path that would reach $v$. 

The first subcase of condition A.1 is expressed as condition 5.3, while the second one is expressed as condition 5.4. Note that the presence of trusted components in the network does not modify \dualrc{}'s connectivity requirement, as the host of a trusted component can still avoid sending messages. 

\emph{($\Leftarrow$)} Two nodes that are neighbors can always reliably communicate (condition 5.1). If condition 5.2 is verified then $v$ will receive $f+1$ disjoint dissemination paths. If condition 5.3 is verified then node $v$  will receive the signature of node $u$. Finally, if condition 5.4 is verified then node $v$ will receive a trusted signature.

\hfill $\OpenBox[1pt]$
\end{proof}

\textbf{\emph{RC-Validity} Verification Algorithm - \dualrcverif{}}
Our correctness verification algorithm relies on conditions 5.1--5.4 to maintain the
set of nodes $S$ that are guaranteed to be able to deliver a
broadcaster's message in the presence of up to $f$ faulty nodes.
Initially, $S$ only contains the broadcaster. \dualrcverif{} then
successively attempts to extend $S$ by identifying nodes that will
satisfy one of the conditions NA.1, NA.2, A.1, or A.2 until all nodes have
been added to it (in which case reliable communication is proven), or
until it cannot grow anymore (in which case there exists a set of $f$
nodes that, if faulty, would prevent some nodes from authenticating each other's messages). Checking
\dualrc{}'s correctness on a given graph requires executing our
algorithm for each possible broadcasting node.
    
\section{Related Work}
\label{sec:relatedwork}

\textbf{Reliable communication in the global fault model.} 
Several works studied the reliable communication problem assuming authenticated links in partially connected but unknown networks.
Dolev showed that correct processes can communicate reliably in the presence of $f$ Byzantine nodes if, and only if, the network graph is $(2f{+}1)$-vertex connected~\cite{dolev1981unanimity}. 
Their algorithm, \du, assumes an unknown network topology and disseminates a message along all possible dissemination paths. A variant of this algorithm leverages a routed dissemination of messages in known networks.
Bonomi et al. proposed several optimizations that drastically improve the performance of \du by reducing the overall number of messages~\cite{bonomi2019multi}.
In the authenticated process model, processes can generate a signature for each message they wish to send and rely on network flooding to reach every possible destination. In this work, we consider a static network with the global fault model and leverage both the \du and network flooding-based reliable communication algorithms.   

\textbf{Reliable communication in the local fault model.} 
Koo presented a broadcast algorithm under the $t$-locally bounded fault model~\cite{koo2004broadcast}, which was later coined the Certified Propagation Algorithm (CPA) by Pelc and Peleg~\cite{pelc2005broadcasting}. Maurer and Tixeuil have also defined weaker reliable communication primitives that benefit from higher scalability~\cite{maurer2014byzantine,maurer2014containing}. Recently, Bonomi et al.~\cite{bonomi2024reliable} established conditions for CPA's correctness in dynamic communication networks. 

\textbf{Network requirements for reliable communication.}  
In the authenticated link model, \du requires a $2f{+}1$ connected~\cite{dolev1981unanimity} communication network. This network connectivity requirement in the authenticated process model is lowered to $f{+}1$. 
Several works identified and refined graph-theoretic parameters for CPA to be correct in the local fault model~\cite{Ichimura2010,litsas2013graph,pelc2005broadcasting,tseng2015broadcast}.
In this work, we observe that the presence of trusted processes, along with both authenticated and non-authenticated processes, weakens the network connectivity requirements for reliable communication in the global fault model. Additionally, we identify sufficient conditions for this to occur.

\textbf{Hybrid Byzantine-crash fault model.}
Hadzilacos proposed a model where processes can fail either by halting or failing to send some messages~\cite{hadzilacos1986issues}. 
Thambidurai and Park~\cite{thambidurai1988interactive}, and Meyer and Pradhan~\cite{meyer1991consensus} considered the agreement problem under a hybrid fault model that involves crashing and Byzantine processes. 
Garay and Perry studied reliable broadcast and consensus when up to $f$ nodes are either Byzantine or fail by crashing~\cite{garay1992continuum}. 
Backes and Cachin described a reliable broadcast protocol that considers a hybrid model where up to $f$ Byzantine nodes and $t$ nodes can crash and recover~\cite{backes2003reliable}.

\textbf{Hybrid Byzantine-trusted fault model.}
More recently, Tseng et al. modified CPA to leverage trusted nodes~\cite{tseng2019reliable, tseng2020reliable}.
Abraham et al. described an authenticated multivalued consensus protocol in synchronous networks resilient to a mix of Byzantine and crash faults~\cite{abraham2022authenticated}.
Differently, the hybridization approach equips nodes that run an algorithm with trusted components to improve its resilience~\cite{verissimo2006travelling}.
Several consensus algorithms consider hybrid models that involve authenticated nodes that leverage a trusted component~\cite{correia2002efficient, decouchant2022damysus,Kapitza+al:eurosys:cheapbft:2012, Veronese+al:tc:MinBFT:2013}. 
 Nowadays, trusted components have been implemented with technologies such as Intel SGX~\cite{SGX}, ARM TrustZone~\cite{arm2009arm}, or TPMs~\cite{tpm2016}. As far as we know, our work is the first to leverage a trusted component for reliable communication in a network. 
    
\section{Conclusion}
\label{sec:conclusion}

We proposed novel reliable communication protocols for incomplete network topologies in the global fault model. 
We first extended the state-of-the-art protocols that have respectively been designed for the authenticated process and authenticated link models so that they can tolerate and leverage the possible presence of trusted processes. 
We then showed that these protocols can be combined into \dualrc{}, our novel reliable communication protocol that allows both authenticated and non-authenticated processes to communicate. \dualrc{} leverages the signatures of authenticated processes and the presence of trusted components or trusted processes.
\dualrc{} provides reliable communication over a larger set of network topologies than previous protocols. We also described methods that one can use to verify whether reliable communication is possible using our algorithms on a given network topology and evaluated their complexity. 

\printbibliography

\appendices

\section{Intuition: Leveraging Trusted Nodes in the Authenticated Links Model}
\label{appx:example_dut}

\begin{figure}[H]
    \centering
    \includegraphics[width=0.7\linewidth]{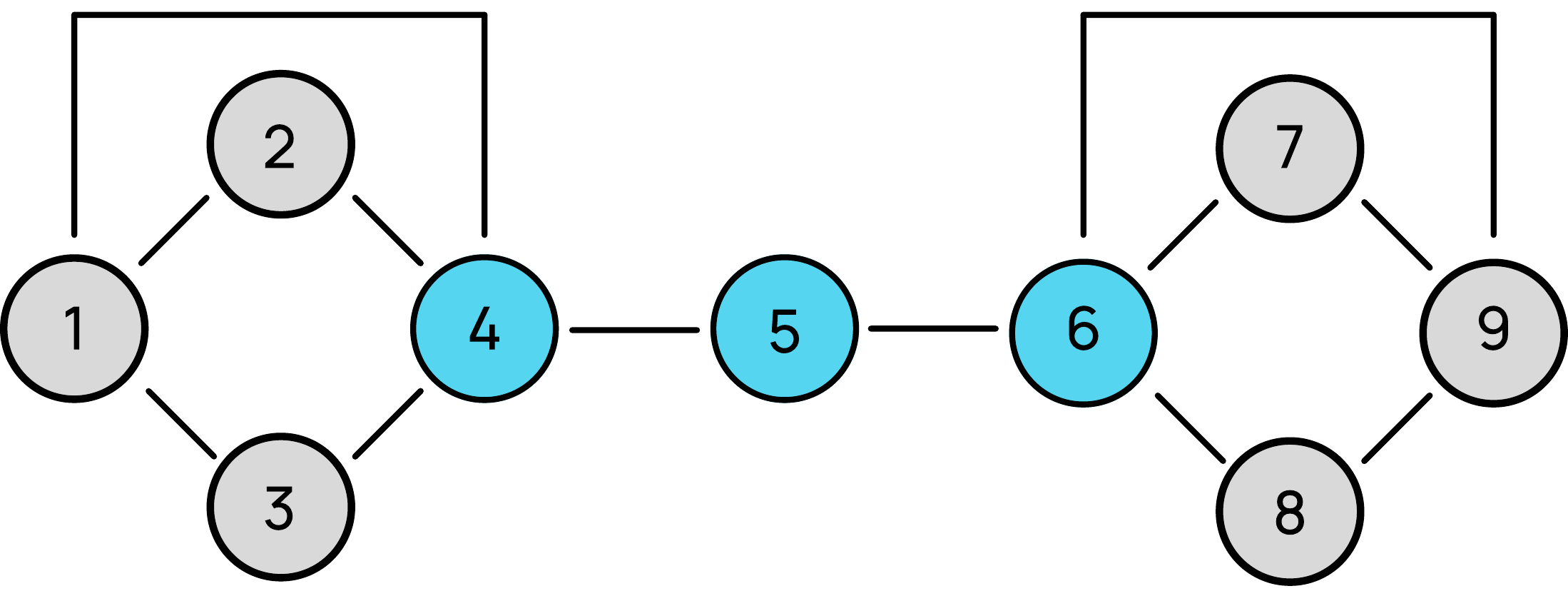}
    \caption{A network of 9 non-authenticated nodes (grey) that includes 3 trusted nodes (blue) where \dut enforces reliable communication while \du does not.}
    \label{fig:naturcExample}
\end{figure}

Fig.~\ref{fig:naturcExample} illustrates a network of non-authenticated nodes that contains 9 nodes, including 3 trusted (i.e., nodes 4, 5, and 6). In this network, reliable communication would be ensured by \dut with at most one faulty node (i.e., $f{=}1$), but not by \du. 

To motivate \dut, it is interesting to observe how a message broadcast by node 1 could be authenticated by node 9, assuming that node 1 is correct. The message would first eventually be authenticated by node 4 (directly or through both untrusted nodes 2 and 3). As node 4 is trusted, it would then vouch for its authenticity to node 5. Later on, node 5 would do the same with node 6, and node 6 would then directly disseminate the message to all its neighbors, including node 9. Since node 6 is trusted, node 9 is able to authenticate the message.
 
This example illustrates the fact that an uninterrupted path of trusted nodes is equivalent to a reliable link, which \dut{} leverages. In particular, a message that only goes through trusted nodes is authenticated. However, \dut{} also relies on the most general observation that trusted nodes can be removed from any position in the paths that \du{} manipulates to decide whether a message can be delivered.

\section{Counter-example: in \dualrc{}, a node must keep forwarding signatures even after it delivers using dissemination paths}
\label{appx:counter_example}

\begin{figure}[H]
\centering
\includegraphics[width=\columnwidth]{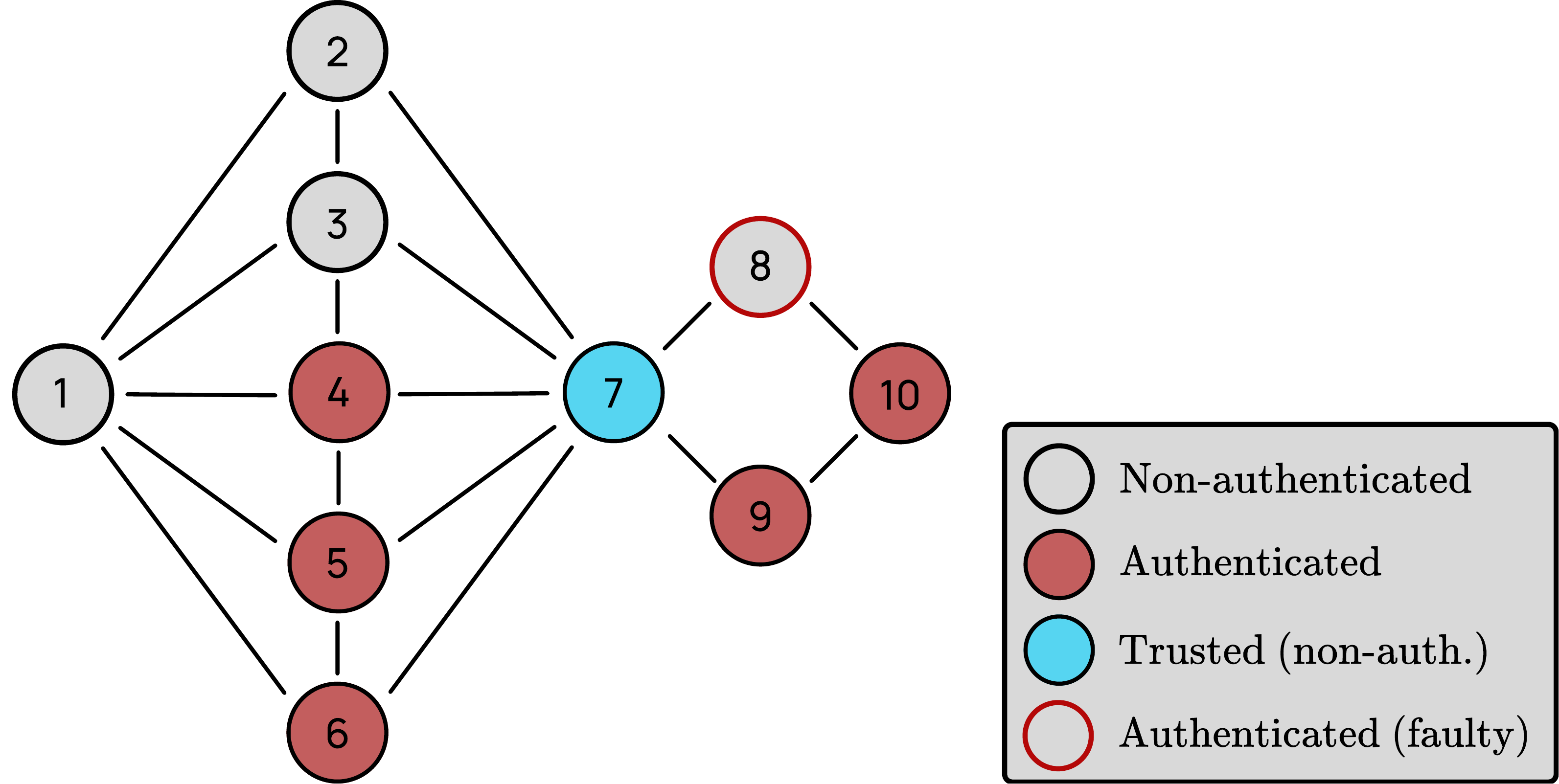}
\caption{Example of a network where not forwarding signatures after delivering a message based on dissemination paths would prevent some nodes from authenticating it. 
}
\label{fig:cex_signs}
\end{figure}

Without knowing the network topology, it is generally not possible for a node that delivers a message based on message paths to stop forwarding all signatures it might receive later.  
Fig.~\ref{fig:cex_signs} shows a graph where doing so can prevent a correct node from authenticating a message from another correct node. In this example, nodes $1$, $2$, $3$, and $7$ are not authenticated and, therefore, do not manipulate signatures (but can relay them) and authenticate messages only based on their dissemination paths. Nodes $4$, $5$, $6$, $8$, $9$ and $10$ are authenticated. Additionally, $7$ is trusted and $8$ is faulty. Node $1$ aims at sending a message $m$ to node $10$. Let us consider the following execution: 

\begin{itemize}
    \item Message $m$ arrives at nodes $2$ and $3$ that both deliver the message and forward it to their neighbors, including node $7$.
    \item Message $m$ arrives at node $7$ through node $2$ with an empty path. At this moment, node $7$ is not able to authenticate $m$ using only path $[2]$, but it forwards $m$ with a modified path.  
    \item Message $m$ arrives at node $7$ through node $3$ with an empty path. Node $7$ authenticates $m$ using paths $[2]$ and $[3]$, and forwards it with an empty path to nodes $8$ and $9$. After delivering $m$, node $7$ ignores all messages it might receive from nodes $4$, $5$, and $6$ that contain their signature on $m$. 
    \item Node $8$ is faulty and does not forward $m$ further.
    \item Node $9$ forwards $m$ with an empty path to node $10$, which cannot authenticate it based solely on $[9]$ and subsequently never authenticates $m$.
\end{itemize}

\noindent Instead, if $7$ keeps forwarding the signatures on $m$ it receives after delivering it based on paths, then $10$ would receive the signatures of nodes $4$, $5$ and $6$ and would authenticate $m$.  \\

\textbf{Remark.} Note that nodes do not have to keep forwarding paths after they have delivered a message using signatures. Indeed, upon delivering a message based on signatures, it is sufficient for a node to forward an empty path to all its neighbors and ignore the paths it might receive afterwards~\cite{bonomi2020boosting}.  

\section{Correctness Verification of \dut{}}
\label{appx:pseudocodes}

\subsection{Complexity of the first correctness verification method for \dut ~--- Max-flow on transformed graph (Alg.~\ref{alg:dolev_u_t_verif_2})}
\label{appx:method1}

The algorithm first requires (i) computing $f(G)$ once, which has
complexity $\mathcal{O}(|V|)$, (ii) computing $g(G,f(G),u_t)$
for each trusted node $u_t$, which has complexity
$\mathcal{O}(|V_t| {\cdot} |V|)$, and computing $g(G,
g(G,f(G),v_t), u_t)$ for each pair of trusted nodes $(u_t, v_t)$,
which has complexity $\mathcal{O}(|V_t|^2 {\cdot} |V|$.

The graph $dG$ has $2|V|$ vertices and contains less than $(2|V|)(2|V|-1)+2|V| = 2{|V|}^2$ edges. The max flow between any two nodes in $dG$ is bounded by $(2f{+}1){\cdot}2(|V|-2 + 1) = (2f{+}1){\cdot}2(|V|-1)$. Computing the max flow in $dgUV$, therefore, is $\mathcal{O}(8{\cdot}f{\cdot}|V|^3)$, which has to be done at most for all pairs of nodes that are not neighbors in $G$ or connected through a path of trusted nodes. The complexity of this procedure is bounded by $\mathcal{O}(8{\cdot}f{\cdot}|V|^5)$.

\subsection{Complexity of the second correctness verification method for \dut ~--- Eliminating unnecessary trusted nodes and checking for \texorpdfstring{$2f{+}1$}{2f+1} vertex-connectivity  (Alg.~\ref{alg:dolev_u_t_verif_1})}
\label{appx:method2}

\textbf{Complexity Analysis.} 
We analyze the complexity of the main steps of Alg.~\ref{alg:dolev_u_t_verif_1} in the following. 

Identifying the edges of $f(G)$ (Alg.~\ref{alg:dolev_u_t_verif_1},
l.~\ref{l:unionfind}) requires computing connected components, which
can be done using a union-find data
structure~\cite{DBLP:journals/cacm/GallerF64}. This data structure is
initialized with sets that each contain one vertex of $G{=}(V,E)$ and
is modified by considering all edges of the trusted nodes in $E$ and
merging the sets that contain their extremities. Let $E_{t} \in E$
contain the edges of $G$'s trusted nodes. Intuitively, a set in the
final union-find data structure contains the nodes that can
communicate using a (possibly empty) path of trusted nodes. Computing
the union-find data structure in a simple way and computing the edges
of $f(G)$ based on it has a linear complexity
$\mathcal{O}(|V|)$.

The neighbors of a node in $f(G)$ are its untrusted neighbors in $G$ 
and the untrusted nodes in the union-find sets that contain at least one of its trusted neighbors in $G$.
Determining the neighbors of all (untrusted) nodes in $f(G)$ based on the union-find data structure therefore has linear complexity $\mathcal{O}(|V|)$.

The number of edge-disjoint paths between two nodes is equal to the maximum flow between them when all edges have unitary capacity (Menger's theorem~\cite{menger1927allgemeinen}). To compute the lowest maximum flow between two nodes in graph $f(G){=}(V_0,E_0)$ one can use the classical Ford-Fulkerson algorithm~\cite{ford1956maximal} by setting all edge capacities to 1. Computing the maximum flow between two nodes in $f(G)$ has a complexity bounded by $\mathcal{O}(|E_0|{\cdot}\mathrm{maxFlow}(f(G)))$, where $\mathrm{maxFlow}(f(G))$ is the maximum flow in $f(G)$.
Evaluating whether two given untrusted nodes can communicate reliably in $f(G)$ (Alg.~\ref{alg:dolev_u_t_verif_1}, l.~\ref{l:untrusted}) therefore has overall complexity $\mathcal{O}(|V {\setminus} V_t|^2 {\cdot} \mathrm{maxFlow}(f(G)))$. Because the number of distinct pairs of edges in $f(G)$ is $\frac{|V{\setminus}V_t|( (|V|{\setminus}V_t| {-} 1)}{2} = \mathcal{O}(|V {\setminus} V_t|^2)$, and because the maximum flow in $f(G)$ is lower than $|V{\setminus}V_t|{-}2$, an upper-bound complexity to check that all pairs of untrusted nodes have a sufficient flow is therefore $\mathcal{O}(|V {\setminus} V_t|^5)$.

Adding a trusted node $u_t$ in $f(G)$ to compute $G_{u_t} {=} g(G,f(G),u_t)$ (Alg.~\ref{alg:dolev_u_t_verif_1}, l.~\ref{l:computeg}) requires merging some connected components of $f(G)$ by considering $u_t$'s edges ($\mathcal{O}(|V|)$ cost), and recomputing the neighbors of each node in $G_{u_t}$ by going through these connected components  ($\mathcal{O}(|V|)$ cost). These steps need to be executed for every trusted node, which results in $\mathcal{O}(|V_t| {\cdot} |V|)$ complexity.

Computing the vertex connectivity between two nodes in $G_{u_t}$ (Alg.~\ref{alg:dolev_u_t_verif_1}, l.~\ref{l:connecutv}) has complexity $\mathcal{O}(|V{\setminus}V_t|^3)$. This step is executed $|V_t| \cdot |V {\setminus} V_t|$ times (all pairs of trusted-untrusted nodes), which leads to an overall complexity of $\mathcal{O}(|V_t| \cdot |V {\setminus} V_t|^4)$. 

Building $G_{u_t,v_t}$ from $G_{u_t}$ (Alg.~\ref{alg:dolev_u_t_verif_1}, l.~\ref{l:gutvt}) has linear complexity $\mathcal{O}(|V|)$. Building such graphs is repeated $\mathcal{O}(|V_t|^2)$ times, for a complexity $\mathcal{O}(|V_t|^2 {\cdot} |V|$.

Evaluating the connectivity between two nodes in $G_{u_t,v_t}$ (Alg.~\ref{alg:dolev_u_t_verif_1}, l.~\ref{l:connecutvt}) has complexity $\mathcal{O}(|V{\setminus}V_t|^3)$ (using Ford-Fulkerson's algorithm~\cite{ford1956maximal}). This step is repeated $\mathcal{O}(|V_t|^2)$ times, which leads to an overall complexity of $\mathcal{O}(|V_t|^2 \cdot |V{\setminus}V_t|^3)$.

The overall complexity of Alg.~\ref{alg:dolev_u_t_verif_1} is bounded by the sum of all terms that appear inside boxes in this section, which is
\begin{multline*} 
\mathcal{O}(
2|V| + |V {\setminus} V_t|^5 + |V_t| {\cdot} |V| +
|V_t| \cdot |V {\setminus} V_t|^4 + |V_t|^2 {\cdot} |V| + |V_t|^2 \cdot |V{\setminus}V_t|^3
).
\end{multline*}

When $|V_t|$ is small, $|V {\setminus} V_t|^5 \approx |V|^5$ dominates. When $|V_t|$ is large, then it is $|V_t|^2 {\cdot} |V| \approx |V|^3$ that dominates. 

\end{document}